\documentclass[11pt]{article}
\usepackage{amsmath,amssymb,amsthm}
\usepackage[numbers,sort]{natbib}
\usepackage{graphicx}
\usepackage{url}

\usepackage[latin1]{inputenc}
\usepackage{fullpage}
\usepackage[colorlinks,linkcolor=blue,citecolor=blue,urlcolor=magenta,linktocpage,plainpages=false]{hyperref}
\usepackage{color}
\usepackage{bm}
\usepackage{times}

\usepackage{tikz}
\usetikzlibrary{shapes.geometric, arrows}

\newtheorem{lemma}{Lemma}
\newtheorem{thm}{Theorem}
\newtheorem{cor}{Corollary}
\newtheorem{conj}{Conjecture}

\newtheorem{definition}{Definition}

\newcommand{\ip}[2]{\langle #1, #2\rangle}
\newcommand{\R}{\mathbb{R}}

\newcommand{\E}{\mathbb{E}}
\newcommand{\ones}{\bm{1}}

\newcommand{\eps}{\varepsilon}

\newcommand{\OPT}{\textsc{OPT}}

\definecolor{mygray}{gray}{0.6}
\newcommand{\remove}[1]{}
\newcommand{\e}{\eps}

\newcommand{\Ko}{\mathcal{K}_{(o)}} %Family of convex sets with 0 in the interior

\newcommand{\cG}{\mathcal{G}}

\newcommand{\conv}{conv}

\newtheorem{manualtheoreminner}{Theorem}
\newenvironment{manualtheorem}[1]{%
  \renewcommand\themanualtheoreminner{#1}%
  \manualtheoreminner
}{\endmanualtheoreminner}

\newcounter{mynotes}
\newcommand\subsetsim{\mathrel{\ooalign{\raise0.2ex\hbox{$\subset$}\cr\hidewidth\raise-0.8ex\hbox{\scalebox{0.9}{$\sim$}}\hidewidth\cr}}}

\DeclareMathOperator{\argmax}{argmax}
\DeclareMathOperator{\argmin}{argmin}

\renewcommand{\cite}[1]{\citep{#1}}
\renewcommand{\citet}[1]{\citep{#1}}

\tikzstyle{process} = [rectangle, minimum width=1cm, minimum height=1cm, text width=3cm, text centered, draw=black]
\tikzstyle{arrow} = [thick,->,>=stealth,double]
\tikzstyle{rlarrow} = [thick,<-,>=stealth,double]
\tikzstyle{dbarrow} = [thick,<->,>=stealth,double]
\tikzstyle{dotarrow} = [->,>=stealth, draw=gray, dashed]

\title{Curvature of Feasible Sets in Offline and Online Optimization}
\author{Marco Molinaro}
\date{}

\begin{document}

\maketitle

\begin{abstract}%
		It is known that the curvature of the feasible set in convex optimization allows for algorithms with better convergence rates, and there has been renewed interest in this topic both for offline as well as online problems.	In this paper, leveraging results on geometry and convex analysis, we further our understanding of the role of curvature in optimization:
		
		\begin{itemize}
			\item We first show the equivalence of two notions of curvature, namely strong convexity and gauge bodies, proving a conjecture of Abernethy et al. As a consequence, this show that the Frank-Wolfe-type method of Wang and Abernethy has accelerated convergence rate $O(\frac{1}{t^2})$ over strongly convex feasible sets without additional assumptions on the (convex) objective function.
			
			\item In Online Linear Optimization, we identify two main properties that help explaining \emph{why/when} Follow the Leader (FTL) has only logarithmic regret over strongly convex sets. This allows one to directly recover a recent result of Huang et al., and to show that FTL has logarithmic regret over strongly convex sets whenever the gain vectors are non-negative. 
			
			\item We provide an efficient procedure for approximating convex bodies by strongly convex ones while smoothly trading off approximation error and curvature. This allows one to extend the improved algorithms over strongly convex sets to general convex sets. As a concrete application, we extend the results of Dekel et al. on Online Linear Optimization with Hints to general convex sets. 
		\end{itemize}
\end{abstract}

%############################################################
%############################################################

	\section{Introduction}
	
	Curvature is one of the most fundamental geometric notions with fascinating connections with many different phenomena. There has been much interest in the influence of curvature on computational and statistical efficiency in optimization and machine learning, with the use of notions of curvature such as \emph{strong convexity} and \emph{gauge bodies} in convex optimization~\citep{levitinPolyak,demyanov,dunn,garberHazan,abernethyEtal}, and \emph{uniform convexity}/\emph{martingale-cotype} (and their dual notions \emph{uniform smoothness}/\emph{martingale-type}) in online and statistical leaning ~\citep{srebroSriTew,regretMartingale,liuLugosiNeuTao,fosterParamFree}. 
	
	Our goal is to better understand the relationship between different notions of curvature and their effect in optimization. We briefly discuss some of the known results in order to point out the specific limitations of current knowledge that we address in this paper. 
	
	\paragraph{Curvature and the Frank-Wolfe method.} Consider a general convex optimization problem
	\begin{align}
		\min ~&f(x)\notag\\
    \textrm{st}\, ~& x \in K, \label{eq:offline}
	\end{align}
	where $f : \R^d \rightarrow \R$ is a convex function and $K \subseteq \R^d$ a convex set. 
	An important procedure for solving such convex programs is the Frank-Wolfe method~\citep{frankWolfe}: in each iteration it solves the linearized version of the problem to obtain a ``direction'' $\tilde{x}_t := \argmin\{\ip{\nabla f(x_{t-1})}{x} : x \in K\}$, where $x_{t-1}$ is the iterate of the previous iteration, and sets the new iterate as $x_t := x_{t-1} + \eta\, (\tilde{x}_t - x_{t-1})$ for some stepsize $\eta > 0$. Because this method only requires optimization of \emph{linear} functions in each iteration, 
and in particular does not require a (non-linear) projection onto the feasible region $K$ as most other methods do, it has  gained much interest in applications to large-scale problems arising in machine learning~\citep{jaggiMatrixCompl,lacosteJulien,Harchaoui2015,robustMatrixRecovery,blockFW,FWapp2}. This method is known to have convergence rate of order $\frac{1}{t}$, i.e., after $t$ iterations it produces a feasible solution of value $+ O(\frac{1}{t})$ compared to the optimal solution, and this is tight in general~\citep{jaggi}. 
	
	However, since the seminal work of Polyak in the 60's, it is known that when the feasible set $K$ is \emph{suitably curved} much better convergence rates are possible~\citep{levitinPolyak,demyanov,dunn,garberHazan}. The common notion of curvature in this context is that of \emph{$\lambda$-strongly convex sets}: for all pairs of points $x,y$, the set needs to contain a large enough ball centered at $\frac{x+y}{2}$. We present a slightly generalized definition that can use another convex body $C$ instead of the Euclidean ball. Recall that given a convex body $C$ with the origin in its interior, its \emph{gauge} is the function $\|\cdot\|_C$ given by
		\begin{align}
		\|x\|_C := \inf \{\lambda > 0 : x \in \lambda C\}. \label{eq:gauge}
		\end{align}
%		\mnote{If only need for symmetric sets, just def gauge for symmetric sets, so can tone down our symmetry assumption later}

	\begin{definition}[Strongly Convex Set~\citep{polyak}] \label{def:SC}
		Let $C$ be a convex body with the origin in its interior. A convex body $K$ is \textbf{\emph{$\lambda$-strongly convex}} with respect to $C$ if for every $x,y \in K$ we have the containment 
		\begin{align}
		\frac{x + y}{2} + \lambda \, \|x-y\|_C^2 \cdot C\subseteq K. \label{eq:defSC}
		\end{align} 
%where $\|\cdot\|_C$ is the gauge function of $C$ (see Section \ref{sec:prelim} for definition). 
	\end{definition}

\begin{figure}[htp]
	\centering		
\includegraphics[width=0.33\textwidth]{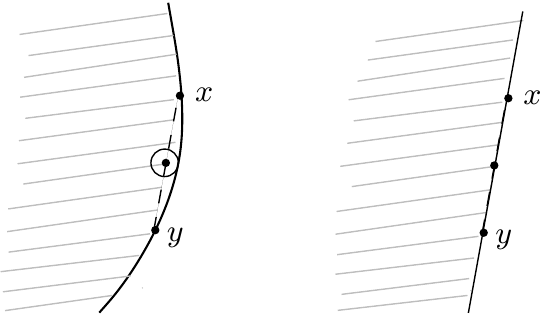}\vspace{-10pt}
	\caption{\small The image on the left shows part of a strongly convex set, with the circle representing a homothetic copy of $C$ centered at $\frac{x+y}{2}$ still contained in the set. The image on the right depicts part of a non-strongly convex set.}
	\label{fig:stronglyConvex}
\end{figure}
		
%	Intuitively, this definition is ruling out flat parts on the boundary of the body (see~\citep{curvedFTL} for a formal connection). Examples of strongly convex sets include $\ell_p$, Schatten $\ell_p$, and group $\ell_{p,s}$ balls for $p,s \in (1,2]$~\citep{garberHazan}.

	%\citet{polyak} considered optimization over strongly convex and showed that under the restrictive assumption that the gradient of the objective function is lower bounded by a positive constant everywhere in the feasible set, the Frank-Wolfe method converges at rate $\frac{1}{e^{\Theta(t)}}$. 
	
	Garber and Hazan \citet{garberHazan} recently showed that as long as the feasible set is strongly convex and the \emph{objective function is a strongly convex function} (see Definition \ref{def:SCF}), the Frank-Wolfe method has accelerated convergence rate $O(\frac{1}{t^2})$. Previous results gave better convergence rates but under additional assumptions.

	However, it seems that even the assumption of strong convexity of the objective function should not be necessary for accelerated convergence rates: the curvature of the feasible set should supersede the curvature of the objective function. In fact, \cite{abernethyAcc} introduced a class of curved convex sets called \emph{gauge sets} and showed that this is indeed the case for them. 
	
%	In a separate line of work,~\cite{abernethyWang} showed that the Frank-Wolfe method can be seen more abstractly as a solution for a convex-concave game where the primal and dual parties play using online learning algorithms, with the regret of these algorithms translating directly to the convergence rate of the Frank-Wolfe-type algorithm. Moreover, in~\citep{abernethyEtal} the authors consider this setup when the feasible set $K$ has the following special property.

	\begin{definition}[Gauge Set~\citep{abernethyEtal}]
		A convex body $K$ with the origin in its interior is a \textbf{\emph{gauge set}} of modulus $G$ with respect to a norm $\|\cdot\|$ if its gauge function squared $\|\cdot\|_K^2$ is a $G$-strongly convex function with respect to  $\|\cdot\|$. 
	\end{definition}
	
	Wang and Abernethy~\citet{abernethyAcc} showed that as long as the feasible region is a gauge set, there is a Frank-Wolfe-type algorithm with convergence rate $O(\frac{1}{t^2})$.
%	\begin{thm}[\cite{abernethyAcc}] \label{thm:abernethyAcc}
%		Consider the problem \eqref{eq:offline}. If $K$ is a \semib{gauge set} and $f$ is strongly smooth, then there is a Frank-Wolfe-type algorithm with convergence rate $O(\frac{1}{t^2})$. Each iteration still only requires minimizing a single linear function over $K$. 
%	\end{thm}
	%The fact this algorithm only uses one linear optimization  oracle call per iteration is quite surprising, since FTRL iteration requires optimizing a \textbf{strongly convex} function (due to the regularizer term $\|\cdot\|_K^2$). This happens because of the integration between this specific regularizer and the feasible set $K$. 	
	%The drawback of this result is that the definition of gauge sets seems less intuitive. 
	While on one hand this result removes the strong convexity requirement of the objective function, on the other it makes a possibly stronger assumption on the feasible set, since the class of gauge sets is contained in that of strongly convex sets~\citep{garberHazan}. However, all standard examples of strongly convex sets such as $\ell_p$, Schatten $\ell_p$, and group $\ell_{p,s}$ balls for $p,s \in (1,2]$, are also gauge sets. This has led Abernethy et al.~\cite{abernethyEtal} to make the following conjecture:
	
	\begin{conj}[\cite{abernethyEtal}] \label{conj:abernethy}
		A convex body $K$ containing the origin in its interior is a gauge set w.r.t. its gauge $\|\cdot\|_K$ if and only if it is strongly convex w.r.t. $K$ itself. 
	\end{conj}
	
	This is one of the gaps in our understanding of curved sets that we address in this paper. Before additional spoilers, we also briefly discuss the role of these sets in \emph{online} optimization.

%###########################################################
%###########################################################
	
	\paragraph{Curvature in online optimization.}	Now consider the \emph{Online Linear Optimization} problem~\citep{OCObook}: A convex set $K$ is given upfront, and objective functions $g_1,g_2,\ldots,g_T$ are revealed one-by-one in an online fashion. In each time step $t$, the algorithm needs to produce a point $x_t \in K$ using the information revealed up to this moment; only \emph{after} that, the  adversary reveals a gain vector $g_t$ from a set $\cG$, and the algorithm receives gain $\ip{g_t}{x_t}$. The goal of the algorithm is to maximize its total gain $\sum_{t = 1}^T \ip{g_t}{x_t}$. Its \emph{regret} for this instance is the missing gain compared to the best fixed action in hindsight: $$\textrm{Regret} := \max_{x \in K} \sum_{t = 1}^T \ip{g_t}{x} - \sum_{t = 1}^T \ip{g_t}{x_t}.$$ We are interested in designing algorithms with provable upper bounds on their worst-case regret. 
	
	This problem, and its generalization with convex objective functions, has a vast literature with applications to a host of areas, from online shortest paths and dynamic search trees~\citep{kalai}, to portfolio optimization~\citep{OLPS15}, to robust optimization~\citep{robustOpt}, and many others. 	
%	 (see for example~\citep{AHK12,OCObook,robustOpt,invOpt}). 
	It is known that as long as the playing set $K$ and the gain vector set $\cG$ are bounded one can obtain order $\sqrt{T}$ regret, and in general this cannot be improved~\citep{OCObook}. On the other hand, when the gain functions are curved (e.g., strongly concave or exp-concave) instead of the linear ones $\ip{g_t}{\cdot}$, it is possible to obtain a much improved order $\log T$ regret~\citep{OCObook}. 
	
	Interestingly,~\citet{curvedFTL} recently showed that one can also obtain this improved order $\log T$ regret when the \emph{playing set $K$} is curved instead; however, they require the additional ``growth condition'' on the gains that $\|g_1 + \ldots + g_t\|_2 \ge t G$ for some $G$ and all $t$. The standard 1-dimensional bad example for Online Linear Optimization shows that an assumption like this growth condition is necessary~\citep{curvedFTL}. It is less clear \textbf{why} this is the case.

	% that the groth
	
%	
%	
% 	There are two interesting features of this result. First, we know FTL in non-curved sets only achieves the trivial regret of order $T$. Moreover, 
% 	
% 	. This result indicates some connection between the role of curvature in the loss function and in the playing set $K$. 
% 	
% 	\medskip
% 	
%	\red{Remove this? On one hand, only used as application of curved, on the other shows one more example of use of curvature (but need to introduce one more notion, etc., may not pay off)}  

%	\blue{Curved sets also have applications in dynamic games (?), see Polovinkin.}
 	
%###########################################################
%###########################################################

	\subsection{Our Results}
	
	Leveraging tools from convex geometry and analysis, we further our understanding of the role of curvature in offline and online optimization. 
	
	\paragraph{Equivalence of strongly convex and gauge sets.} We first observe that Conjecture~\ref{conj:abernethy} of Abernethy et al.~\cite{abernethyEtal} on the equivalence of strongly convex and gauge sets is true.
	
	\begin{thm} \label{thm:gauge}
		Conjecture \ref{conj:abernethy} is true: if the convex body $K \subseteq \R^d$ containing the origin in its interior is $\lambda$-strongly convex with respect to itself,
%\footnote{Using the John Ellipsoid Theorem~\citep{john} on approximations of convex bodies, one can consider strong convexity with respect to other convex bodies, incurring an additional dimension-dependent factor.}
 then $K$ is a gauge set with respect to $\|\cdot\|_K$ with modulus $G = 2\lambda$. 
	\end{thm}
	
	Lemma 3 of \citep{garberHazan} proves the reverse direction, i.e., if $K$ is a gauge set with respect to $\|\cdot\|_K$ with modulus $G$, then $K$ is $\frac{G}{8}$-strongly convex with respect to itself; so the above theorem indeed completes the characterization from Conjecture \ref{conj:abernethy}.

	The main idea of the proof is to use as a stepping stone another classic notion of curvature introduced by Clarkson~\cite{clarkson} in the context of geometry of Banach spaces, namely \emph{2-convexity} of norms (Definition \ref{def:twoConv}). When $\|\cdot\|_K$ is symmetric (i.e., a norm) is is easy to see that $K$ is a gauge set iff $\|\cdot\|_K$ is 2-convex, and it is known that $\|\cdot\|_K$ is 2-convex iff $\|\cdot\|_K^2$ is strongly convex with respect to $\|\cdot\|_K$~\cite{lindTza}, proving the theorem in the symmetric case. However, without symmetry it is not clear how to carry out the usual proofs of this last equivalence. Instead, we proceed though a longer route that uses the duality between 2-convexity/2-smoothness of gauges and strong convexity/strong smoothness of functions. The advantage is that can more easily prove the equivalence of 2-smoothness of $\|\cdot\|_K^*$ and the strong convexity of $(\|\cdot\|_K^*)^2$ w.r.t. $\|\cdot\|_K^*$ based on a second-order differential argument, that is, using the Hessian $\nabla^2 (\|\cdot\|_K^*)^2$. This seems to bypass the issues of asymmetry because the quadratic form $y \mapsto y^\top (\nabla^2 f(x)) y$ is always symmetric.
	
% The result then follows directly form the fact that $K$ is strongly convex with respect to itself iff its gauge function $\|\cdot\|_K$ is 2-convex with same moduli (Lemma \ref{lemma:SCUC}), and the fact that if $\|\cdot\|_K$ is 2-convex then the squared gauge $\|\cdot\|_K^2$ is strongly convex (for example, Lemma 1.e.10 of \cite{lindTza}, restated in the appendix in Lemma \ref{lemma:lindTza}).  
	
	In addition to clarifying the relationship between these two notions of curvature, it shows that the Frank-Wolfe-type algorithm of \citep{abernethyAcc} is the first to achieve accelerated rates under the standard notion of strong convexity of the feasible set \emph{without any additional assumption on the objective function} (besides convexity).
		
	\begin{cor}
		Consider the problem \eqref{eq:offline}. If $K$ is a strongly convex body, then the Frank-Wolfe-type algorithm of \citep{abernethyAcc} has convergence rate $O(\frac{1}{t^2})$.\footnote{The $O(\cdot)$ hides other parameters that influence the convergence of the algorithm, such as the modulus of strong smoothness of the objective function (which is always finite over compact sets).}
	\end{cor}
	
%	The main idea for the proof is to use as a stepping stone another classic notion of curvature introduced by~\cite{clarkson} in the context of geometry of Banach spaces, namely \emph{2-convexity} of norms (Definition \ref{def:twoConv}). The result then follows from observing the simple fact that $K$ is strongly convex with respect to itself iff its gauge function $\|\cdot\|_K$ (equation \eqref{eq:gauge}) is 2-convex (with same moduli), and the fact that if $\|\cdot\|_K$ is 2-convex then the squared gauge $\|\cdot\|_K^2$ is strongly convex. The latter is proved, for example, in~\cite{beauzamy} \mnote{Add Lindenstrauss Tzatzifiri Classic Banach Space II?} for centrally symmetric bodies, and the proof did not keep track of the modulus of strong convexity of $\|\cdot\|_K^2$, but these issues can be easily handled (respectively by symmetrizing only one component in the proof, and appealing to the local nature of strong convexity). %We provide these details in Appendix~\ref{app:conj}.
	
	%Despite the extensive literature on strongly convex sets (\cite{polovinkin,balashovRepovs,balashovRepovs,weber,weber2,vial,surveyStronglyConvex} and many more), this result does not seem to have been considered. 
	
%#############################################################
		
	\paragraph{Online Linear Optimization on curved sets.} Next, we identify two main properties that help explaining \textbf{why} curvature helps in online optimization.
	
	\begin{thm}[Informal principle] \label{thm:principle}
		In Online Linear Optimization, the improved regret guarantees observed in \citep{curvedFTL} for strongly convex playing sets $K$ (attained by the Followed the Leader algorithm) stems from 
		\begin{align*}
			\textrm{Partial Lipschitzness of the support function of $K$} + \textrm{no-cancellation of the gain vectors}.
		\end{align*}	
	\end{thm} 
	
	This principle is described and developed in detail in Section \ref{sec:FTL} (see Lemmas \ref{lemma:lipSphere} and \ref{lemma:regretLip} for some formal statements). But at a high level, the first property is intimately related to the stability of the Follow the Leader (FLT) algorithm, which is known to control its regret. However, this Lipschitzness only holds away from the origin. That is why the additional no-cancellation property of the gain vectors is required: it steers the iterates of FTL away from the origin.
	
	This principle gives a simple and clean proof of the $O(\log T)$ regret result of ~\citep{curvedFTL}, where this no-cancellation is achieved through the linear growth assumption on the partial sums of the gain vectors. As another illustration of this principle, we use it to obtain a new result which shows that FTL has logarithmic regret over strongly convex sets when the gain vectors are \emph{non-negative}, without any additional growth assumption (Theorem~\ref{thm:FTLnew}).

\newcommand{\thmFTLnew}{Consider the Online Linear Optimization problem with playing set $K$ and gain set $\cG$. If $K$ is $\lambda$-strongly convex with respect to a norm $\|\cdot\|$ and all vectors $\cG$ are non-negative,\footnote{That is, $\cG \subseteq \R^d_+$. We note that the proof directly generalizes to the case when $\R^d_+$ is replaced by an arbitrary pointed cone.} then FTL has regret at most 
		\begin{align*}
		\frac{C \cdot M}{\lambda} \cdot \log T,
		\end{align*}
		where $M := \max_{g \in \cG} \|g\|$ and $C$ only depends on $\|.\|$.}

	\begin{thm} \label{thm:FTLnew}
		\thmFTLnew
	\end{thm}
	
	Again this should be contrasted with the standard regret $O(\sqrt{T})$ for non-strongly convex sets, which cannot be improved even for non-negative gains (even in the special case of prediction with expert advice~\cite{OCObook}).

	 %Note that the non-negativity assumption is just another way of achieving the no-cancellation property. 

%	Finally, we note that the algorithm for OLO with hints of Nika and company also obtains $\log T$ regret for the wider class of strongly convex sets (it is wider because of Lemma \ref{xx}).
%	
%	\begin{thm}
%	\ldots
%	\end{thm}
%
%	Not only this allows us to generalize Nika's result to more general curved playing sets, but it allows to get better parameters when things are better adapted to a norm. For example as a corollary we have the following. (apply with the constant hint $c_t = (\frac{1}{\sqrt{n}},\ldots, \frac{1}{\sqrt{n}})$. 
%	
%	
%	\begin{cor}
%	\ldots.
%	\end{cor}
	
	\paragraph{Making a convex body curved.} In order to extend results obtained for curved set to general sets, we also give an efficient way of transforming an arbitrary convex body $K$ into a curved one while controlling both its curvature as well as its distance to the original set. We use $B(r)$ to denote the Euclidean ball of radius $r$ of appropriate dimension. 
		
	% 	Again the authors show that if the playing set $K$ is ``suitably curved'', one can obtain order $\log T$ regret. The notion of curvature now is the one introduced by~\cite{clarkson} in the geometry of Banach spaces. 

	\begin{thm} \label{thm:curving}
		Consider a convex body $K$ and suppose $B(r) \subseteq K \subseteq B(R)$. Then for all $t \in [0,1]$, there is a convex body $K_t$ with the following properties:
		\vspace{-4pt}
		\begin{enumerate}
			\itemsep0em
			\item (Approximation) $K_t \subseteq K \subseteq \sqrt{1 + \big(\big(\frac{R}{r}\big)^2 - 1\big)\,t^2} \cdot K_t$ \vspace{-4pt}
%			\item (Curvature) The gauge $\|\cdot\|_{K_t}$ is 2-convex with modulus $D = \frac{t^2}{8}$
			\item (Curvature) $K_t$ is $\frac{t^2}{8}$-strongly convex with respect to itself
			\item (Efficiency) Given access to a weak optimization oracle for $K$, weak optimization over $K_t$ can be performed in time that is polynomial in $r,R$, and the desired precision $\delta$ (see Definition \ref{def:weakOpt}). 
		\end{enumerate}
	\end{thm}

	We note that this construction smoothly interpolates between the original set $K$ when $t = 0$ and the inscribed ball $B(r)$ when $t = 1$, and the guarantees interpolate with no loss at the endpoints.
	
	 The starting element for this construction is again the equivalence between strong convexity of sets and 2-convexity of their gauge functions. Based on this, the construction of $K_t$ uses the ``Asplund averaging'' technique for combining (2-convex) norms into a 2-convex one~\citep{asplund}: $K_t$ is defined by setting its gauge to be $\|\cdot\|_{K_t} := ((1-t^2) \|\cdot\|_K^2 + t^2 \|\cdot\|_{B(r)}^2)^{1/2}$. Equivalently, $K_t$ can be defined based on the so-called \emph{$L_2$ addition} of the (scaled) polars of $K$ and $B(r)$, an operation introduced by Firey~\cite{firey}. In fact, in to order show that one can optimize over $K_t$ in polynomial time, we resort to an equivalent characterization of this operation given by Lutwak et al.~\cite{lutwak2}.	

	As a concrete example of application, we consider the problem of \emph{Online Linear Optimization with hints} and show how Theorem \ref{thm:curving} allows us to port the low regret algorithm of Dekel et al.~\citep{nika}, designed for strongly convex playing sets, to general playing sets, at the expense of a small multiplicative regret.

%	Moreover, we hope our approximation by strongly convex sets ones will find other applications, given the usefulness of these sets:\red{Carefull: isn't smoothness the DUAL?} the notions of martingale-type/(2,D)-smooth norms, equivalent to 2-convexity, were used by~\cite{srebroSriTew} on the optimality of Online Mirror Descent, by \cite{liuLugosiNeuTao} on algorithmic
%stability, by \cite{regretMartingale} on the equivalence of regret inequalities and martingale bounds, and by \cite{fosterParamFree} on parameter-free online learning. \mnote{Can we get application from Lugosis paper?}	

%
%	\subsection{Structure of the paper}
%		
%	In order to reduce context-switching first we prove the more structural results Theorems \ref{thm:gauge} and \ref{thm:curving}, and leave the principle from Theorem \ref{thm:principle} to be described and developed in detail in the last section. 	
	
%############################################################
%############################################################
%############################################################
%############################################################

	\section{Preliminaries} \label{sec:prelim}
	
We need some basic notions from convex analysis, for which we refer to the book~\citep{HUL}. 
	
	\begin{definition} \label{def:SCF}
		A convex function $f : \R^d \rightarrow \R$ is \emph{$G$-strongly convex} with respect to a norm $\|\cdot\|$ if for all $x,y$ and all $\alpha \in [0,1]$	
		\begin{align*}
			f(\alpha x + (1-\alpha) y) \,\le\, \alpha f(x) +  (1-\alpha) f(y) \,-\, G \cdot \alpha (1-\alpha) \|x-y\|^2,
		\end{align*}
 and is \emph{$G$-strongly smooth} with respect to $\|\cdot\|$ if for all $x,y$ and all $\alpha \in [0,1]$	
		\begin{align*}
			f(\alpha x + (1-\alpha) y) \,\ge\, \alpha f(x) +  (1-\alpha) f(y) \,-\, G \cdot \alpha (1-\alpha) \|x-y\|^2,
		\end{align*}		
		
	\end{definition}
		
	\paragraph{Set operations, gauge and support functions.}
		Recall that $B(r)$ denotes the Euclidean ball of radius $r$ in the appropriate dimension depending on the context. Given a set $A$ and a scalar $\lambda \in \R$ we define $\lambda A := \{\lambda\, a : a \in A\}$, and given two sets $A,B$ we define their Minkowski sum $A + B := \{a + b : a \in A,\, b \in B\}$ and their difference $A - B := \{x \in A : x + B \subseteq A\}$ (so $A - B(r)$ has the interpretation of the points ``deep inside'' $A$). A set $A$ is \emph{(centrally) symmetric} if $A = -A$. By a convex body we mean a compact convex set with non-empty interior. We use $\Ko^d$ to denote the set of all convex bodies in $\R^d$ with 0 in their interior; we work almost exclusively with convex bodies in such position. 
		
		Given such a convex body $K \in \Ko^d$, its \emph{support function} is $$\sigma_K(\ell) := \max_{x \in K} \ip{x}{\ell},$$ and recall that its \emph{gauge} is $\|x\|_K := \inf \{\lambda > 0 : x \in \lambda K\}$.
		
			Gauge functions are generalization of norms: every norm $\|\cdot\|$ is the gauge of its unit norm ball $\{x : \|x\| \le 1\}$, and gauge functions satisfy all properties of norms (as listed below) other than symmetry, which holds iff the convex body is centrally symmetric. We need the following standard facts about these operators that can be readily verified.% (see also Section C.3 of~\cite{HUL}).
	
	\begin{lemma} \label{lemma:gaugeSupp}
	 For convex bodies $K, K'$ with the origin in their interior, we have the following:
	 \vspace{-4pt}
	 \begin{enumerate}
	 	\itemsep0em
	 	\item (level set) $K$ is precisely the set of points $x$ satisfying $\|x\|_K \le 1$
	 	\item (positive homogeneity) For every scalar $\lambda \in \R_+$, $\|\lambda x\|_K = \lambda \|x\|_K$
		\item (subadditivity) $\|x + y\|_K \le \|x\|_K + \|y\|_K$
		\item (inclusion) $K' \subseteq K$ iff $\|\cdot\|_{K'} \ge \|\cdot\|_K$ pointwise, and iff $\sigma_{K'}(\cdot) \le \sigma_K(\cdot)$ pointwise 
%\red{Only about $\sigma$: Theorem C.3.3.1 of \cite{HUL}}
	 	\item (scaling of body) For all $\lambda \in \R_+$, $\|\cdot\|_{\lambda K} = \frac{1}{\lambda} \|\cdot\|_{K}$, and $\sigma_{\lambda K}(\cdot) = \lambda \sigma_K(\cdot)$ pointwise. 
	\end{enumerate}
	\end{lemma}
	
	\paragraph{Polarity.}	The \emph{polar} of a convex body $K \in \Ko^d$ is the convex body $$K^\circ := \{y : \ip{x}{y} \le 1,\,\forall x \in K\}.$$ We will also need the following properties of polars.

	\begin{lemma}	\label{lemma:polar2}
		For convex bodies $K, K'$ with the origin in their interior, we have the following:
		\vspace{-4pt}
		\begin{enumerate}
		\itemsep0em
	 	\item (polar involution) $(K^\circ)^\circ = K$ %\red{Equation C.3.2.7}
	 	\item (polar order reversal) $K \subseteq K'$ iff $K^\circ \supseteq K'$ 
	 	\item (duality of functionals) $\|x\|_K = \sigma_{K^\circ}(x)$	
	 	\item (Euclidean balls) For all $r > 0$ we have $B(r)^\circ = B(\frac{1}{r})$.
	 \end{enumerate}
	\end{lemma} 
	
	For a gauge $\|\cdot\|$, we use $\|y\|_{\star} := \sup_{\|x\| \le 1} \ip{x}{y}$ to denote its dual gauge. By definition, we have the generalized Cauchy-Schwarz inequality: 
	\begin{align}
		\ip{x}{y} \le \|x\| \|y\|_{\star}. \label{eq:genCS}
	\end{align}
	Note that since $\|\cdot\|_{K^\circ} = \sigma_K$, we see that $\|\cdot\|_{K^\circ}$ is the dual gauge of $\|\cdot\|_K$. In particular, by involution of polarity we have $\|\cdot\|_{\star \star} = \|\cdot\|$. 

	Also, when $\|\cdot\|_K$ is differentiable at $x$ we have that
	\begin{align}
		\|\nabla \|x\|_K\|_{K^\circ} \le 1,   \label{eq:gradNorm}
	\end{align}
	since $$\nabla \|x\|_K = \nabla \sigma_{K^\circ}(x) = \argmax_{y \in K^\circ} \ip{y}{x}$$ and hence $\nabla \|x\|_K \in K^{\circ}$.

%#########################################################

	\paragraph{Fenchel conjugate.} The \emph{Fenchel conjugate} of a function $f : \R^d \rightarrow \R$ is the function $f^\star$ given by 
	\begin{align*}
		f^\star(z) := \sup_{x \in \R^d} \Big(\ip{x}{z} - f(x)\Big).
	\end{align*}
	
	We have the following relationship between the gauge functions of polar sets and Fenchel conjugacy (see for example equation (1.49) on page 55 of \cite{schneider}).

	\begin{lemma} \label{lemma:commSquare}
		For any gauge $\|\cdot\|$ over $\R^d$ we have $(\frac{1}{2} \|\cdot\|^2)^\star = \frac{1}{2}\|\cdot\|_{\star}^2$.
	\end{lemma}

%#########################################################

	\paragraph{Subgradients.} Given a convex function $f : \R^d \rightarrow \R$, its \emph{subdifferential} at $x$, denoted by $\partial f(x)$, is the set of all vectors $g \in \R^d$ that give an underestimation of the function, namely $$f(y) \ge f(x) + \ip{g}{y-x} \qquad\qquad \textrm{for all } y \in \R^d.$$ A vector $g \in \partial f(x)$ is called a \emph{subgradient}. Furthremore, if $f$ is differentiable at $x$ then $\partial f(x)$ is the singleton set consisting of the gradient $\nabla f(x)$. 

%#########################################################
%#########################################################
%#########################################################
%#########################################################

	\section{Equivalence of Strongly Convex and Gauge Bodies} \label{sec:equiv}

	In this section we prove that strongly convex sets are gauge sets (Theorem \ref{thm:gauge}). The argument follows the high-level chain depicted in the figure below:% Figure \ref{fig:proof}.

	\vspace{4pt}
  \begin{figure}[!h]
	\centering
	\scalebox{0.95}{
	\begin{tikzpicture}[node distance=2.7cm]
		\node (SCSet) [process] {$K$ is strongly convex};
		\node (UCGauge) [process, right of=SCSet, xshift=2.8cm] {$\|\cdot\|_K$ is $(2,D)$-convex};
		\node (USGauge) [process, right of=UCGauge, xshift=3cm] {$\|\cdot\|_{K^\circ}$ is $(2,\frac{1}{16D})$-smooth};
		\node (SSSquared) [process, below of=USGauge] {$\|\cdot\|^2_{K^\circ}$ is strongly smooth};
		\node (SCSquared) [process, below of=UCGauge] {$\|\cdot\|^2_{K}$ is strongly convex};
		\draw [dbarrow] (SCSet) -- node[anchor=north] {} node [anchor=south] {{\footnotesize Lemma \ref{lemma:SCUC}}} (UCGauge);
		\draw [dbarrow] (UCGauge) -- node[anchor=north] {} node [anchor=south] {{\footnotesize Lemma \ref{lemma:CSduality}}}  (USGauge);
		\draw [arrow] (USGauge) -- node[anchor=east] {{\footnotesize Lemma \ref{lemma:smoothSmooth}}} node [anchor=west] {} (SSSquared);
		\draw [arrow] (SSSquared) -- node[anchor=north] {} node [anchor=south] {\footnotesize Lemma \ref{lemma:fenchelSSSC}} (SCSquared);
	\end{tikzpicture}
	}
%	\caption{Structure of the proof of Theorem \ref{thm:gauge}. The proofs in the dashed arrow assume symmetry of the set~$K$.}
%	\label{fig:proof}
\end{figure}
	\vspace{4pt}

	 The main stepping stone is another classic notion of curvature in Banach spaces~\citep{clarkson}.
	
 	\begin{definition}[Uniform convexity] \label{def:twoConv}
		A gauge function $\|.\|_K$ is \textbf{\emph{$(2,D)$-convex}} if for all $x,y$ satisfying $\|x\|_K \le 1$ and $\|y\|_K \le 1$ we have 
		\begin{align}
			\bigg\|\frac{x + y}{2}\bigg\|_K \le 1 - D \|x-y\|^2_K. \label{eq:twoConv}
		\end{align}
	\end{definition}
	
	Notice that for $x,y$ as above, the subadditivity of gauges gives that $\|\frac{x+y}{2}\|_K \le \frac{\|x\|_K}{2} + \frac{\|y\|_K}{2} \le 1$; thus, 2-convexity gives an improvement depending on how far $x$ and $y$ are from each other. As an example, the Euclidean norm is $(2,\frac{1}{8})$-convex, and this modulus $\frac{1}{8}$ is best possible (see \cite{lindTza}, page 63). 

	The first step for proving Theorem \ref{thm:gauge} is  noticing that a gauge $\|\cdot\|_K$ is 2-convex iff the set $K$ is strong convex with respect to itself. Despite the extensive literature on strongly convex sets (see the survey \citep{surveyStronglyConvex}), we could not find a reference for this result. We present its simple proof for completeness. 
	
		\begin{lemma} \label{lemma:SCUC}
		A convex body $K \in \Ko^d$ is $\lambda$-strongly convex with respect to itself iff its gauge $\|\cdot\|_K$ is $(2,\lambda)$-convex. 
	\end{lemma} 	
	
	\begin{proof}
	$(\Rightarrow)$	Take $x,y$ such that $\|x\|_K, \|y\|_K \le 1$, so $x,y \in K$. Let $m = \frac{x + y}{2}$ and $t = \|x - y\|_K$. Using the $\lambda$-strong convexity of $K$ at $m$, we have that the point $m + \lambda t^2 \frac{m}{\|m\|_K}$ belongs to $K$, and hence $$1 \ge \Big\|m + \lambda t^2 \tfrac{m}{\|m\|_K}\Big\|_K = \left(1 + \frac{\lambda t^2}{\|m\|_K}\right) \|m\|_K = \|m\|_K + \lambda t^2.$$ Thus, $\|m\|_K \le 1 - \lambda t^2$, proving the $2$-convexity of $\|\cdot\|_K$.
	
	\medskip
	\noindent $(\Leftarrow)$ Take $x,y \in K$ with $\|x - y\|_K \ge t$, so by assumption $\|\frac{x+y}{2}\|_K \le 1 - D t^2$. Then for any $w \in D t^2 \cdot K$ we have by triangle inequality $\|\frac{x+y}{2} + w\| \le 1$, i.e., this point belongs to $K$. This means that $\frac{x+y}{2} + D t^2 \cdot K$ is contained in $K$. Thus, $K$ is $\lambda$-strongly convex with respect to itself with $\lambda = D$.  	
	\end{proof}
	
	For the next step, we need the following dual version of uniform convexity that measures how ``flat'' the set $K$ on the most ``curved'' point of its boundary.	

 	\begin{definition}[Uniform smoothness] \label{def:twoSmooth}
		A gauge function $\|.\|_K$ is \textbf{\emph{$(2,D)$-smooth}} if for all $x$ satisfying $\|x\|_K = 1$ and for all $y \in \R^d$ we have 
		\begin{align}
			\frac{\|x + y\|_K + \|x - y\|_K}{2} \le 1 + D \|y\|_K^2. \label{eq:twoSmooth}
		\end{align}
	\end{definition}
	
	As an example, the Euclidean norm is $(2,\frac{1}{2})$-smooth, and this modulus $\frac{1}{2}$ is best possible~(see \cite{lindTza}, page~63). 
	
	The connection with the previous step is that indeed 2-convexity and 2-smoothness are duals of each other. The standard proof of this fact for norms (for example, that of Proposition 1.7 of Chapter IV of \cite{smoothnessRenorming}) goes through unchanged for gauges.
	
	\begin{lemma} \label{lemma:CSduality}
		A gauge $\|\cdot\|$ is $(2,D)$-convex iff its dual gauge $\|\cdot\|_{\star}$ is $(2,\frac{1}{16 D})$-smooth.
	\end{lemma} 	
	
	The next crucial link on the chain is that 2-smoothness of a gauge implies strong smoothness of the squared gauge. This follows from known facts for symmetric gauges (i.e., norms), but unfortunately here the proofs \emph{do not} directly generalize to asymmetric gauges. This is our main technical contribution in order to prove Theorem \ref{thm:gauge}, and in order to keep the high-level chain of the argument its proof is deferred to Section~\ref{sec:smoothSmooth}. 

	\begin{lemma} \label{lemma:smoothSmooth}
		If $\|\cdot\|$ is a $(2,D)$-smooth gauge, then the square gauge $\|\cdot\|^2$ is a $(2D + 1)$-strongly smooth function with respect to $\|\cdot\|$.
	\end{lemma}	
	
	The final link in the chain is that the duality between strongly smooth and strongly convex functions also holds with respect to gauges; standard proofs such as that of Proposition 2.6 of \cite{azePenot} works with minor changes, but we provide a proof for convenience.
	
	\begin{lemma} \label{lemma:fenchelSSSC}
		If a function $f$ is $D$-strongly smooth with respect to a gauge $\|\cdot\|$, then its Fenchel conjugate $f^\star$ is $\frac{1}{4D}$-strongly convex with respect to the dual gauge $\|\cdot\|_{\star}$. 
	\end{lemma}	
	
	\begin{proof}
		Fix $x,y \in \R^d$ and $\alpha \in [0,1]$. From the definition of Fenchel conjugate, for any $u, w \in \R^d$ (let $v = u-w$)
		\begin{align*}
			\alpha f^\star(x) + (1-\alpha) f^\star(y) &\ge \alpha \bigg[ \ip{u}{x} - f(u) \bigg] + (1-\alpha) \bigg[ \ip{v}{y} - f(v) \bigg]\\
%			&= \alpha \ip{u}{x} + (1-\alpha) \ip{v}{y} - \bigg[\alpha f(u) + (1-\alpha) f(v) \bigg]\\
			&\ge \alpha \ip{u}{x} + (1-\alpha) \ip{v}{y} - \bigg[f(\alpha u + (1-\alpha) v) + D \alpha (1-\alpha) \|u-v\|^2 \bigg]\\
			&= \ip{\alpha u + (1-\alpha)v}{\alpha x + (1-\alpha) y} - f(\alpha u + (1-\alpha) v) \\
			&~~~~~~~~~~+ \alpha (1-\alpha) \bigg[ \ip{u-v}{x-y} - D \|u-v\|^2 \bigg]\\
			&= \bigg[\ip{u - (1-\alpha)w}{\alpha x + (1-\alpha) y} - f(u - (1-\alpha) w)\bigg] \\
			&~~~~~~~~~~+ 2 D \alpha (1-\alpha) \bigg[ \ip{w}{\tfrac{x-y}{2D}} - \tfrac{1}{2} \|w\|^2 \bigg],
		\end{align*}
		where the second inequality follows from the assumption that $f$ is $D$-strongly smooth w.r.t. $\|\cdot\|$. Taking a supremum over $u$ the first bracket in the right-hand side it becomes $f^\star(\alpha x + (1-\alpha) y)$ (from the definition of $f^\star$ and the fact $u \mapsto u + (1-\alpha)w$ ranges over all $\R^d$), and then taking supremum over $w$ the second bracket becomes $(\frac{1}{2}\|\frac{x-y}{2D}\|^2)^\star$, which equals $\frac{1}{2}\|\frac{x-y}{2D}\|_{\star}^2$ from Lemma \ref{lemma:commSquare}. This gives that
		\begin{align*}
		\alpha f^*(x) + (1-\alpha) f^*(y) \ge f^*(\alpha x + (1-\alpha) y) + \frac{\alpha (1-\alpha)}{4D} \|x-y\|^2_{\star},
		\end{align*}
		and hence $f^\star$ is $\frac{1}{4D}$-strongly convex w.r.t. $\|\cdot\|_{\star}$. 
	\end{proof}

	\begin{proof}[Proof of Theorem \ref{thm:gauge}]
	All the above lemmas put together directly imply Theorem \ref{thm:gauge}. More precisely, consider a set $K$ that is $\lambda$-strongly convex with respect to itself. Due to the remark right after Definition~\ref{def:twoConv} and Lemma~\ref{lemma:SCUC}, we know $\lambda \le \frac{1}{8}$, a fact that will be useful later. 
	
	To simplify the notation let $\|\cdot\| = \|\cdot\|_{K}$. Chaining Lemmas \ref{lemma:SCUC}, \ref{lemma:CSduality}, and \ref{lemma:smoothSmooth}, we get that the squared dual gauge $(\|\cdot\|_{\star})^2$ is a $(\frac{1}{8 \lambda} + 1)$-strongly smooth function with respect to $\|\cdot\|_{\star}$. Using the previous lemma, this implies that $((\|\cdot\|_\star)^2)^\star$  is a $G$-strongly convex function w.r.t. $\|\cdot\|_{\star \star} = \|\cdot\|$ with $G = (\frac{1}{8 \lambda} + 1)^{-1} \ge \lambda$, where the last inequality follows from the observation $\lambda \le \frac{1}{8}$. However, from Lemma \ref{lemma:commSquare}  $(\frac{1}{2}(\|\cdot\|_\star)^2)^* = \frac{1}{2}\|\cdot\|^2_{\star \star} = \frac{1}{2}\|\cdot\|^2$. This means that $\|\cdot\|^2$ is a $2\lambda$-strongly convex function w.r.t. $\|\cdot\|$. This concludes the proof of the theorem. 
	\end{proof}

%############################################################
%############################################################
%############################################################

	\subsection{Proof of Lemma \ref{lemma:smoothSmooth}} \label{sec:smoothSmooth}

	We now turn to the proof of Lemma \ref{lemma:smoothSmooth}.  We first consider the case when $\|\cdot\|_K$ is twice differentiable on $\R^d \setminus \{0\}$, and then handle the general case with a (non-trivial) approximation argument. We note that $(2,D)$-smoothness of a gauge does not guarantee twice differentiability, e.g. $\|x\| = (x_1^2 + (x_2^4 + x_3^4)^{1/2})^{1/2}$ is $(2,D)$-smooth but not twice differentiable at $(1,0,0)$.

%############################################################
%############################################################

	\subsubsection{Twice differentiable gauges}

Assume throughout this section that $\|\cdot\|_K$ is $(2,D)$-smooth and twice differentiable on $\R^d \setminus \{0\}$. We use $\nabla^2 f$ to denote the Hessian of a function $f$.

We start with the fact that to show strong smoothness of a function it suffices to upper bound its Hessian; this fact is classical for strong smoothness with respect to the Euclidean norm, and we just sketch that a standard proof also holds for the gauge case (and with the origin excluded).
	
	\begin{lemma} \label{lemma:hessImplSmooth}
		If $f : \R^d \rightarrow \R$ is twice differentiable on $\R^d \setminus \{0\}$ and $$y^\top (\nabla^2 f(x)) y \le D \|y\|_K^2$$ for all $x \neq 0$ and all $y$, then $f$ is $\frac{D}{2}$-strongly smooth with respect to $\|\cdot\|_K$. 
	\end{lemma}
	
	\begin{proof}
	Fix $u,v \in \R^d$ and $\lambda \in [0,1]$. Assume for now that $0$ does not belong to the convex hull $\conv(\{u,v\})$ of $u$ and $v$. Then by the mean-value version of Taylor's theorem, there are non-zero vectors $x,x'$ such that 
	\begin{align*}	
		f(u) &= f(\lambda u + (1-\lambda) v) + \ip{\nabla f(\lambda u + (1-\lambda) v)}{(1-\lambda) (u-v)} + \frac{1}{2} (1-\lambda)^2 (u-v)^\top (\nabla^2 f(x)) (u-v)\\
		f(v) &= f(\lambda u + (1-\lambda) v) + \ip{\nabla f(\lambda u + (1-\lambda) v)}{\lambda (v-u)} + \frac{1}{2} \lambda^2 (u-v)^\top (\nabla^2 f(x')) (u-v),
	\end{align*}
	noting that we have $(u-v)$ in the last terms of both equations. Upper bounding $\nabla^2 f(x)$ and $\nabla^2 f(x')$ using the assumption of the lemma, and adding $\lambda$ times the first equation plus $(1-\lambda)$ times the second equation, yields
	\begin{align*}	
		\lambda f(u) + (1-\lambda) f(v) \le f(\lambda u + (1-\lambda) v) + \frac{D}{2} \lambda (1-\lambda) \|u-v\|_K^2.
	\end{align*}	
	By continuity of $f$, this inequality holds even if $0 \in \conv(\{u,v\})$ (e.g. by applying the inequality to $u'= u + t \ones$ and $v' = v + t \ones$ and taking $t\rightarrow 0$). This proves that $f$ is $\frac{D}{2}$-strongly smooth with respect to $\|\cdot\|_K$.
	\end{proof}
	
	Thus, to prove the desired property that $\|\cdot\|_K^2$ is strongly smooth w.r.t. $\|\cdot\|_K$, it suffices to upper bound its hessian. For $x,y \in \R^d$, define the function $f_{x,y}(t) := \|x + t y\|_K$. By the twice differentiability of $\|\cdot\|_K$, we have $f''_{x,y}(0) = y^\top (\nabla^2 \|x\|_K) y$ as long as $x \neq 0$. Moreover, if $\|x\|_K = 1$ we have
	\begin{align*}
		f''_{x,y}(0) = \lim_{t\rightarrow 0} \frac{f_{x,y}(t) + f_{x,y}(-t) - 2 f_{x,y}(0)}{t^2} = \lim_{t\rightarrow 0} \frac{\|x + t y\|_K + \|x- t y\|_K - 2}{t^2} \le 2D \|y\|_K, 
	\end{align*} 
%the formula can be checked via l'opital, or see https://math.stackexchange.com/questions/690060/definition-of-second-derivative-as-a-limit/690068
	where the last inequality follows from the $(2,D)$-smoothness of $\|\cdot\|_K$. Together these give the following bound for the hessian of $\|\cdot\|_K$ (but not yet of $\|\cdot\|_K^2$): for all $x$ such that $\|x\|_K = 1$ and all $y$
	\begin{align}
		\,y^\top (\nabla^2 \|x\|_K) y \le 2 D \|y\|_K^2. ~~~~~\label{eq:hess}
	\end{align}

%	
%	
%	{\color{gray} To start, take $x$ such that $\|x\|_K = 1$ and $y \in \R^n$. Using the second order expansion of $\|\cdot\|_K$ at $x$ in the directions $y$ and $-y$ and noticing that the first-order terms cancel out, we get 
%	%
%	\begin{align*}
%		\frac{\|x+y\|_K + \|x - y\|_K}{2} = \|x\|_K + \frac{1}{2}\,y^\top (\nabla^2 \|x\|_K) y + o(\|y\|^2_2).
%	\end{align*}
%	Since $\|\cdot\|_K$ is $(2,D)$-smooth, the left-hand side is at most $\|x\|_K + D \|y\|_{K}^2$ and hence
%	%
%	\begin{align*}
%		\,y^\top (\nabla^2 \|x\|_K) y \le 2 D \|y\|_K^2 + o(\|y\|^2_2).
%	\end{align*}	
%	Since this holds for all $y$, replacing it by $\e y$, multiplying throughout by $\frac{1}{\e^2}$, and taking $\e \rightarrow 0$ we see that 
%		%
%	\begin{align}
%		\,y^\top (\nabla^2 \|x\|_K) y \le 2 D \|y\|_K^2. \label{eq:hess}
%	\end{align}	
%	}

	Next, we apply the chain rule and see that the hessian of the squared gauge $\|\cdot\|_K^2$ has the following form
	\begin{align}
		\nabla^2 \|x\|_K^2 = 2 (\nabla \|x\|_K)\,(\nabla \|x\|_K)^\top + 2 \|x\|_K (\nabla^2 \|x\|_K) \label{eq:hessSquared}
	\end{align}
	for all $x \in \R^d \setminus \{0\}$. Since $\|\cdot\|_K$ is positively homogeneous of degree 1, Euler's theorem on homogeneous functions gives that $\nabla \|\cdot\|_K$ is positively homogeneous of degree 0 and $\nabla^2 \|\cdot\|_K$ is positively homogeneous of degree -1, namely
	\begin{align*}
	\nabla^2 \|\alpha x\|_K = \frac{1}{\alpha} \nabla^2 \|x\|_K ~~~~~~~\forall \alpha > 0.
	\end{align*}
	Letting $\tilde{x} := \frac{x}{\|x\|_K}$ and applying this observation to \eqref{eq:hessSquared} we obtain 
	\begin{align*}
		\nabla^2 \|x\|_K^2 = 2 (\nabla \|x\|_K)\,(\nabla \|x\|_K)^\top + 2 (\nabla^2 \|\tilde{x}\|_K) ~~~~~~~\forall x \in \R^d \setminus \{0\}.
	\end{align*}
	Then for every $x \in \R^d \setminus \{0\}$ and $y \in \R^d$ we have
	\begin{align*}
		y^\top (\nabla^2 \|x\|_K^2) y &= 2 \ip{\nabla \|x\|_K}{y}^2 + 2 y^\top (\nabla^2 \|\tilde{x}\|_K) y \le 2 \|y\|_{K}^2 + 4D \|y\|_K^2,
	\end{align*}
	where the inequality follows from $\ip{\nabla \|x\|_K}{y} \le \|\nabla\|x\|_K\|_{K^{\circ}}\cdot  \|y\|_K \le \|y\|_K$ (using \eqref{eq:genCS} and \eqref{eq:gradNorm}), and from \eqref{eq:hess}.
	
	Applying Lemma \ref{lemma:hessImplSmooth} then gives that $\|\cdot\|_K^2$ is $(2D+1)$-strongly smooth with respect to $\|\cdot\|_K$. This proves Lemma \ref{lemma:smoothSmooth} when $\|\cdot\|_K$ is twice differentiable.

%############################################################
%############################################################

	\subsubsection{General gauges}

	The idea is to approximate $\|\cdot\|_K$ with a gauge $\|\cdot\|_\e$ that is twice differentiable on $\R^d \setminus \{0\}$. While several such approximations are available, the difference is that $\|\cdot\|_\e$ also needs to preserve the $(2,D)$-smoothness of $\|\cdot\|_K$. For that we will use the following approximation result, which is a direct consequence of Theorem 3.4.1 of \cite{schneider} (applied to the support function of the polar $K^\circ$), which we prove next preserves $(2,D)$-smoothness. Recall that a function is $C^\infty$ if it has continuous derivatives of all orders. 
	
	\begin{lemma}
		Consider a gauge function $\|\cdot\|_K$. Then for every $\e \in (0,1)$ there is a gauge function $\|\cdot\|_\e$ with the following properties:
		
		\begin{itemize}
			\item $\|\cdot\|_{\e} : \R^d \rightarrow \R$ is a $C^\infty$ function over $\R^d \setminus \{0\}$
			\item $\|x\|_K (1- \e) \le \|x\|_\e \le (1+\e)\|x\|_K$ for all $x \in \R^n$
			\item There is a random variable $Z \in \R^n$ with atomless distribution such that $\|x\|_{\e} = \E_Z \big\|x + Z\,\|x\|_2\big\|_K$ and $\|Z\|_K \le \e$ with probability 1.  
		\end{itemize}
	\end{lemma}
%M: More detail:
%Apply the theorem to sigma_{K^circ}. Then tilde{sigma_{K^circ}}(x) = E[sigma_{K^circ}(x + Z |x|)]. 
% Then T is the map that maps K^circ to A := T K^circ where sigma_A = tilde{sigma_{K^circ}}. 
%Point (c) of the theorem guarantees that delta(K^circ, T K^circ) = SUP_{x : |x|_2 = 1} |sigma_{K^circ}(x) - tilde{sigma_{K^circ}}(x)| <= e R, which is equivalent to SUP_{x : |x|_2 = 1} | |x|_K - |x|_e| <= e R. So for all x such that |x|_2 = 1 we have |x|_e = |x|_K \pm e*R*|x|_2. 
%By homogeneity this holds for all x. By renanimg e we have (1 \pm e) |x|_K	

	\begin{lemma}[Approximation preserves 2-smoothness]
		If $\|\cdot\|_K$ is $(2,D)$-smooth with modulus $D > 0$, then the gauge $\|\cdot\|_\e$ from the previous lemma is $(2,D (1 + O(\e)))$-smooth (where the asymptotics $O(\e)$ is with $\e \rightarrow 0$).
	\end{lemma}
	
	\begin{proof}
		Before starting, we note that since $K$ is compact and contains the origin in its interior there are positive scalars $M, N < \infty$ such that:
		\begin{gather*}
			\|-x\|_K \le M \|x\|_K  ~~~~\textrm{for all $x \in \R^n$}\\
			\frac{1}{N} \|x\|_K \le \|x\|_2 \le N \|x\|_K~~~~\textrm{for all $x \in \R^n$}.
		\end{gather*}
		We assume WLOG that $\e$ is sufficiently small compared to $N$ and $M$. To simplify the notation we use $\|\cdot\|$ to denote $\|\cdot\|_K$, keeping in mind that $\|\cdot\|$ is not symmetric. 
		
		Using homogeneity of $\|\cdot\|$, $(2,D)$-smoothness of $\|\cdot\|$ is equivalent to its homogenized version (which does not require $\|x\|=1$)
		\begin{align*}
			\frac{\|x+y\| + \|x-y\|}{2} &\le \|x\| + D \cdot \frac{\|y\|^2}{\|x\|} &\textrm{for all $x \neq 0$, all $y$},
		\end{align*} 
		which using $u=x+y$ and $v=x-y$ is equivalent to 
		\begin{align}
			\frac{\|u\| + \|v\|}{2} &\le \bigg\|\frac{u+v}{2}\bigg\| + D \cdot \frac{\|\frac{u-v}{2}\|^2}{\|\frac{u+v}{2}\|} &\textrm{for all $u,v$ such that $u+v \neq 0$}. \label{eq:equivSmooth}
		\end{align}
		Thus, $\|\cdot\|$ satisfies this inequality, and our goal is to prove that $\|\cdot\|_{\e}$ satisfies this inequality with $D(1+O(\e))$ replacing $D$. 
		
		So consider $u,v$ such that $u+v \neq 0$. 
		
		\paragraph{Case 1: $\|u + v\| \le \sqrt{D}\, \|u-v\|$.} This implies that $\|\frac{u+v}{2}\| \le D \|\frac{u-v}{2}\|^2/ \|\frac{u+v}{2}\|$. Using this (on the starred inequality below) and the approximation properties of $\|\cdot\|_\e$ we get
		\begin{align*}
			\frac{\|u\|_\e + \|v\|_{\e}}{2} \le (1 + \e) \frac{\|u\| + \|v\|}{2} &\stackrel{\eqref{eq:equivSmooth}}{\le} (1+\e) \bigg\|\frac{u+v}{2}\bigg\| + (1+\e) D \cdot \frac{\|\frac{u-v}{2}\|^2}{\|\frac{u+v}{2}\|} \\
			&\stackrel{(\star)}{\le}  (1-\e) \bigg\|\frac{u+v}{2}\bigg\| + (1+3\e) D \cdot \frac{\|\frac{u-v}{2}\|^2}{\|\frac{u+v}{2}\|} \\
			&\le \bigg\|\frac{u+v}{2}\bigg\|_\e + \frac{(1+3\e) (1+\e)}{(1-\e)^2} D \cdot \frac{\|\frac{u-v}{2}\|^2_\e}{\|\frac{u+v}{2}\|_\e}.
		\end{align*}
		The multiplier on the last term is $D(1-O(\e))$, and so $\|\cdot\|_\e$ satisfies the modified version of \eqref{eq:equivSmooth} as desired.
		
%############################################################

	\paragraph{Case 2: $\|u + v\| > \sqrt{D}\, \|u-v\|$.} Let $Z$ be the random variable such that $\|x\|_\e = \E_Z \|x + Z \|x\|_2\|$, and define $U = u + Z \|u\|_2$ and $V = v + Z \|v\|_2$. Applying \eqref{eq:equivSmooth} to $U$ and $V$ in each scenario where $U + V \neq 0$ (which happens with probability 1 since $Z$ is atomless) 
		\begin{align}
			\frac{\|u\|_\e + \|v\|_{\e}}{2} &\,=\, \E \bigg[\frac{\|U\| + \|V\|}{2}\bigg] \le \E\, \bigg\|\frac{U+V}{2}\bigg\| + D \cdot \E\bigg[ \frac{\|\frac{U-V}{2}\|^2}{\|\frac{U+V}{2}\|}\bigg]. \label{eq:case2}
		\end{align}
	We bound the terms in the right-hand side. To upper bound the first term, by subadditivity
	\begin{align}
		\E\, \bigg\|\frac{U+V}{2}\bigg\| &\,\le\, \E\, \bigg\|\frac{u+v}{2} + Z\,\|\tfrac{u + v}{2}\|_2\bigg\| \,+\, \E\,\bigg\|Z \Big(\tfrac{\|u\|_2 + \|v\|_2}{2} - \|\tfrac{u+v}{2}\|_2\Big)\bigg\| \notag\\
		&\,\le\, \bigg\|\frac{u+v}{2}\bigg\|_{\e} \,+\, \E\|Z\| \cdot  \Big(\tfrac{\|u\|_2 + \|v\|_2}{2} - \|\tfrac{u+v}{2}\|_2\Big)\notag\\
		&\,\le\, \bigg\|\frac{u+v}{2}\bigg\|_{\e} \,+\, \frac{\e}{2}  \cdot \frac{\|\frac{u-v}{2}\|^2_2}{\|\frac{u+v}{2}\|_2} \notag\\
		&\,\le\, \bigg\|\frac{u+v}{2}\bigg\|_{\e} \,+\, \frac{\e N^3}{2} \cdot \frac{\|\frac{u-v}{2}\|^2}{\|\frac{u+v}{2}\|}    ,  \label{eq:UBUplusV}
	\end{align}
	where the third inequality follows from $\|Z\| \le \e$ and the fact the Euclidean norm is $(2,\frac{1}{2})$-smooth. We also have the following lower bound that holds with probability 1:
	\begin{align}
		\bigg\|\frac{U+V}{2}\bigg\| \,\ge\, \bigg\|\frac{u+v}{2}\bigg\| - \bigg\|-Z\cdot\tfrac{\|u\|_2 + \|v\|_2}{2}\bigg\| &\,\ge\, \bigg\|\frac{u+v}{2}\bigg\| - \e M \bigg( \bigg\|\frac{u+v}{2}\bigg\|_2 + \frac{1}{2} \cdot \frac{\|\frac{u-v}{2}\|^2_2}{\|\frac{u+v}{2}\|_2}\bigg)\notag \\
		&\,\ge\, (1-\e M N) \bigg\|\frac{u+v}{2}\bigg\| -   \frac{\e M N^3}{2}  \cdot \frac{\|\frac{u-v}{2}\|^2}{\|\frac{u+v}{2}\|} \notag\\
		&\,\ge\, (1-\e M N - \tfrac{\e M N^3}{2D}) \bigg\|\frac{u+v}{2}\bigg\|, \label{eq:LBUplusV}
	\end{align}
	where in the second inequality we use that $\|-Z\| \le M \|Z\| \le \e M$ and again the fact that the Euclidean norm is $(2,\frac{1}{2})$-smooth,  and in the last inequality we use the assumption of Case 2. Finally, we upper bound
	\begin{align}
		\bigg\|\frac{U - V}{2}\bigg\| \le \bigg\|\frac{u-v}{2}\bigg\| + \e\bigg|\frac{\|u\|_2-\|v\|_2}{2}  \bigg| \le \bigg\|\frac{u-v}{2}\bigg\| + \e \,\bigg\|\frac{u - v}{2}  \bigg\|_2 \le (1 + \e N) \bigg\|\frac{u-v}{2}\bigg\|. \label{eq:UBUminusV}
	\end{align}
	
	Plugging the bounds from inequalities \eqref{eq:UBUplusV}-\eqref{eq:UBUminusV} on \eqref{eq:case2} gives 
		\begin{align*}
			\frac{\|u\|_\e + \|v\|_{\e}}{2} ~\le~ \bigg\|\frac{u+v}{2}\bigg\|_\e + D (1 + O(\e)) \cdot \frac{\|\frac{u-v}{2}\|^2}{\|\frac{u+v}{2}\|} ~\le~ \bigg\|\frac{u+v}{2}\bigg\|_\e + D (1 + O(\e)) \cdot \frac{\|\frac{u-v}{2}\|^2_\e}{\|\frac{u+v}{2}\|_\e}, \\
		\end{align*}
		showing $\|\cdot\|_\e$ satisfies the modified version of \eqref{eq:equivSmooth} as desired. This concludes the proof of the lemma.		
	\end{proof}
	
	We finally prove Lemma \ref{lemma:smoothSmooth} without any differentiability assumption. Consider $\|\cdot\|_K$ and the smooth approximation $\|\cdot\|_\e$ given above. Since $\|\cdot\|_\e$ is twice differentiable and $(2, D (1+O(\e)))$-smooth, we know from the previous section that $\|\cdot\|_\e^2$ is a $(D (1+O(\e)) + 1)$-strongly smooth function w.r.t. $\|\cdot\|_\e$, i.e. for each fixed $u,v \in \R^d$ and $\alpha \in [0,1]$
	\begin{align*}
		\|\alpha u + (1-\alpha) v\|_\e^2 \ge \alpha \|u\|_\e^2 + (1-\alpha) \|v\|_\e^2 - \Big(D (1+O(\e)) + 1\Big)\,\alpha (1-\alpha) \|u-v\|_\e^2.
	\end{align*} 
	Taking $\e \rightarrow 0$ shows that $\|\cdot\|^2$ is $(2D+1)$-strongly smooth with respect to $\|\cdot\|$, finishing the proof of the lemma.

%############################################################
%############################################################
%############################################################
%############################################################

		\section{Online Linear Optimization on Curved Sets} \label{sec:FTL}	
		
		The goal of this section is to develop the informal principle stated in Theorem \ref{thm:principle}. We briefly recall the Online Linear Optimization (OLO) problem described in the introduction: a convex body $K$ (playing set) is given upfront; in each time step the algorithm first produces a point $x_t \in K$ using the information obtained thus far, sees a gain vector vector $g_t$, and obtains gain $\ip{g_t}{x_t}$. The goal is to minimize the regret against the best fixed action:
		\begin{align*}
			\textrm{Regret} := \max_{x \in K} \sum_{t =1}^T \ip{g_t}{x} - \sum_{t=1}^T \ip{g_t}{x_t}.
		\end{align*}
		We are interest in the case where $K$ is strongly convex. 
		
		Follow the Leader (FTL) is arguably the simplest algorithm for this problem, being simply greedy in the previous gain vectors: letting $s_t := g_1 + \ldots + g_t$, the algorithm at time $t$ chooses an action 
		\begin{align}
		x_t \in \argmax_{x \in K} \ip{s_{t-1}}{x} \tag{FTL}
		\end{align}
		($x_1$ is chosen as an arbitrary point in $K$). It is well-known that whenever FTL is stable, namely actions $x_t$ and $x_{t-1}$ on consecutive times are ``similar'', it obtains good regret guarantees; in fact, this is the basis for the analysis of most OLO algorithms. More precisely, Lemma 2.1 of~\citep{shalevSurvey} gives the following.
		
		\begin{lemma} \label{lemma:FTLbasic}
			The regret of FTL is at most $\sum_{t =1}^T \ip{g_t}{x_{t+1} - x_t}$.
		\end{lemma}  
	
	Unfortunately in general FTL can be quite unstable: For example, consider the instance $K = [-1,1]^2$, with gain sequence $g_1 = (1, 0.01)$ and for $t \ge 2$ the gains $g_t$ alternate between $(1,-0.1)$ and $(1,0.1)$. Even though the gain vectors are very similar across time steps, the actions of FTL alternate between $(1,1)$ and $(1,-1)$, being extremely unstable. In addition, its regret is $\Omega(T)$, which up to constants is worst possible.
	
	 However, the intuition is when $K$ is ``curved'', we should have $x_{t+1} \approx x_t$, as long as the directions of $s_{t+1}$ and $s_t$ are similar, see Figure \ref{fig:FTL}.a. 		More formally, notice that $x_t$ is the optimizer of the support function $\sigma_K(s_{t-1}) \stackrel{\textrm{def.}}{=} \max_{x \in K} \ip{s_{t-1}}{x}$, and because of that it is a subgradient of it: $x_t \in \partial \sigma_K(s_{t-1})$.
%	 ; this is intuitive, since $\sigma_K(s)$ is the max of all linear functions $\{s \mapsto \ip{s}{x}\}_{x \in K}$, the ``right linearization'' of $\sigma_K$ at $s$ should be given by a linear function $s \mapsto \ip{s}{x^*}$ satisfying $\sigma_K(s) = \ip{s}{x^*}$. 	 
	 In addition, if $K$ is strongly convex, then $\sigma_K$ is differentiable everywhere except the origin, and hence $x_t = \nabla \sigma_K(s_{t-1})$ as long as $s_{t-1} \neq 0$~\citep{diffSupport}, see Figure \ref{fig:FTL}.b.

\begin{figure}[h]
	\centering		
\includegraphics[width=0.85\textwidth]{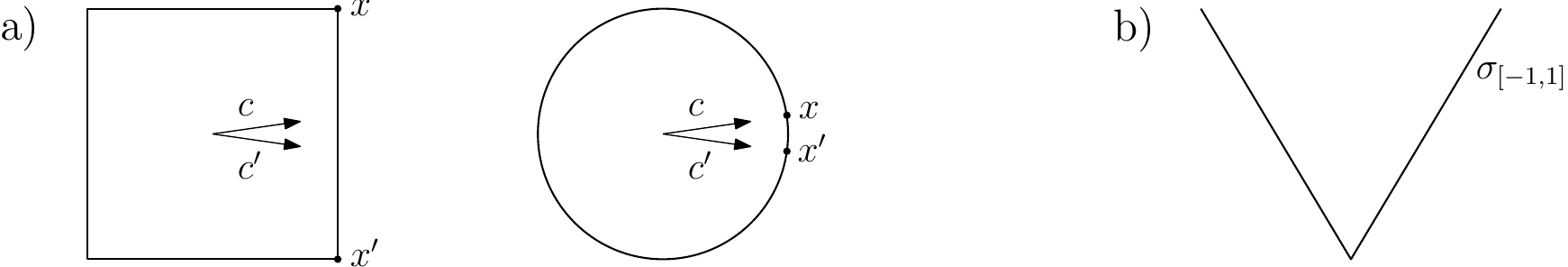}
	\caption{\small a) Points $x$ and $x'$ represent the maximizers of the linear functions $c$ and $c'$ over $K$, respectively. When $K$ is not strongly convex, like the cube, $x$ and $x'$ can be far from each other even if the directions $c,c'$ are similar. When $K$ is strongly convex, like the ball, the maximizers get closer as the directions get closer. b) The support function of the set $[-1,1]$ is $\sigma_{[-1,1]}(s) = \max\{-s,s\}$. The gradients exists and are Lipschitz everywhere except the origin. Notice $\nabla \sigma_{[-1,1]}(s) \in \{-1,1\}$ is the vertex $x^*$ of the set that achieves $\sigma_{[-1,1]}(s) = s \cdot x^*$. Notice that the gradient $\nabla \sigma_{[-1,1]}$ is not stable around the origin.}
	\label{fig:FTL}
\end{figure}
		
	Thus, the FTL stability requirement $x_{t+1} \approx x_t$ is equivalent to $\nabla \sigma_K(s_t) \approx \nabla \sigma_K(s_{t-1})$, namely stability of the gradient of the support function. A big problem is that since $\sigma_K$ is never differentiable at the origin, gradients are not stable around there. %For example, if $K = [-1,1]$, $g_1 = -0.5$ and the other $g_t$'s alternate $\pm 1$, we have $\nabla \sigma_K(s_t)$ alternating $\pm 1$; notice $s_t = \pm 0.5 \approx 0$ ($[-1,1]$ is even strongly convex). 

	But when $K$ is strongly convex, this is the only problem: $\nabla \sigma_K$ is Lipschitz \textbf{over the unit sphere}. This fact has been rediscovered several times~\citep{vial,polovinkin,balashov,abernethyEtal}; for example, this is Lemma 2.2 of~\cite{balashov}.
	
	\begin{lemma}[Lipschitz gradients over the sphere] \label{lemma:lipSphere}
		If $K \subseteq \R^d$ is $\lambda$-strongly convex with respect to a norm $\|\cdot\|$, then for all $u,v$ with $\|u\| = \|v\| = 1$ we have
		\begin{align*}
			\|\nabla \sigma_K(u) - \nabla \sigma_K(v)\|_{\star} \le \frac{1}{4\lambda} \|u- v\|.
		\end{align*}
%		that is, the map $\nabla \sigma_K : (\R^d, \|\cdot\|) \rightarrow (\R^d, \|\cdot\|_{\star})$ is $\frac{1}{4\lambda}$-Lipschitz over the sphere. 
	\end{lemma}
		
	Just using this limited ``sphere-Lipschitz'' property (and Lemma \ref{lemma:FTLbasic}) we get a generic upper bound on the regret of FTL on strongly convex sets.\footnote{This is similar to the conclusion of Proposition 2 plus inequality (6) of~\citep{curvedFTL}, but arguably with a simpler and more transparent proof.}

	\begin{lemma}[FTL regret from sphere-Lipschitz] \label{lemma:regretLip}
		If $K \subseteq \R^d$ is such that the gradient of its support function $\sigma_K$ satisfies the Lipschitz gradient condition of Lemma \ref{lemma:lipSphere}, then the regret of FTL is at most 
		\begin{align*}
		 \frac{1}{2\lambda} \sum_t \frac{\|g_t\|^2}{\|s_t\|},
		\end{align*}	
		as long as $s_t \neq 0$ for all $t$. 
	\end{lemma}	
	
	\begin{proof}
	To simplify the notation we drop the subscript from $\sigma_K$. From Lemma \ref{lemma:FTLbasic} and the generalized Cauchy-Schwarz inequality \eqref{eq:genCS}, the regret of FTL is at most
	\begin{align}
	\textrm{regret} \le \sum_{t =1}^T \|g_t\|\|x_{t+1} - x_t\|_{\star} = \sum_{t =1}^T \|g_t\|\|\nabla \sigma(s_t) - \nabla \sigma(s_{t-1})\|_{\star}. \label{eq:FTL1}
	\end{align}
	We upper bound this dual norm. By positive homogeneity of $\sigma$ we have $\nabla \sigma(s_t) = \nabla \sigma(\frac{s_t}{\|s_t\|})$, so Lemma~\ref{lemma:lipSphere} implies
		\begin{align}
			\|\nabla \sigma(s_t) - \nabla \sigma(s_{t-1})\|_{\star} &= \bigg\|\nabla \sigma\left(\tfrac{s_t}{\|s_t\|}\right) - \nabla \sigma\left(\tfrac{s_{t-1}}{\|s_{t-1}\|}\right)\bigg\|_{\star} \le \frac{1}{4\lambda} \left\|\frac{s_t}{\|s_t\|} - \frac{s_{t-1}}{\|s_{t-1}\|} \right\| \label{eq:FTL2}.
		\end{align}
	We claim that the norm on the right-hand side is at most $2\frac{\|g_t\|}{\|s_t\|}$. To see this, since $s_t = s_{t-1} + g_t$ we can use triangle inequality to upper bound it by 
	\begin{align}
			\frac{\|g\|}{\|s_t\|} + \bigg\|\frac{s_{t-1}}{\|s_t\|} - \frac{s_{t-1}}{\|s_{t-1}\|} \bigg\|  = \frac{\|g_t\|}{\|s_t\|} + \bigg|\frac{\|s_{t-1}\|}{\|s_t\|} - 1 \bigg| \le \frac{\|g_t\|}{\|s_t\|} + \frac{\|g_t\|}{\|s_t\|} = 2\frac{\|g_t\|}{\|s_t\|},
		\end{align}
		where in the first equation we used the manipulation $\|\alpha u\| = |\alpha| \|u\| = | \alpha \|u\||$ valid for any scalar $\alpha$ and vector $u$, and in the inequality we again used triangle inequality to get $\|s_{t-1}\| \in \|s_t\| \pm \|g_t\|$, which implies $\frac{\|s_{t-1}\|}{\|s_t\|} \in 1 \pm \frac{\|g_t\|}{\|s_t\|}$. Combining the displayed equations gives the result.  
	\end{proof}
	
	Now we just need to control the denominator of this expression, namely to bound $s_t$ away from the origin. This is what we refer to as the \textbf{``no-cancellation''} property. We consider two incarnations of this property.

%############################################################
%############################################################
	
	\subsection{Applications: logarithmic regret}
	
	\paragraph{No-cancellation via growth condition on $s_t$.} 	 We can guarantee the desired no-cancellation by assuming that there is $G$ such that $\|s_t\| \ge t G$ for all $t$, recovering the $\log T$ regret from Theorem \ref{thm:curvedFTLHuang} of~\cite{curvedFTL}.
	
	\begin{thm}[\cite{curvedFTL}] \label{thm:curvedFTLHuang}
		Consider the Online Linear Optimization problem with playing set $K$ and gain set $\cG$. If $K$ is 
%a $C^2$ body\footnote{A convex body is $C^2$ if its boundary is a twice continuously differentiable manifold, see Sec. 2.5 of~\cite{schneider}.} and is
 $\lambda$-strongly convex
 %\footnote{The original statement in the paper uses the notion of principal curvatures instead, which is equivalent, see Proposition 4 of~\cite{curvedFTL}.}
  w.r.t. the Euclidean ball and  the gain vectors satisfy the growth condition $\|g_1 + \ldots + g_t\|_2 \ge t G$ for some $G$ and all $t$, then 
%  . Finally, let $M := \max_{\ell \in \cL} \|\ell\|_2$ and assume the sequence of loss functions satisfies $\|\ell_1 + \ldots + \ell_t\|_2 \ge t L$ for some $L$ and all $t$. Then 
the algorithm Follow the Leader has regret at most $$\frac{M^2}{2\lambda G} (1 + \log T),$$ where $M := \max_{g \in \cG} \|g\|_2$.
	\end{thm}

% precisely the assumption in. With the development above, we directly recover this result (and extend it to arbitrary norms): under this assumption, the regret of FTL is at most $\frac{1}{2\lambda}\sum_{t = 1}^T \frac{\|g_t\|^2}{tG} \le \frac{\max_{g_t} \|g_t\|^2}{2 \lambda G} \log T$. 	
	
%############################################################	

	\paragraph{No-cancellation via non-negative gain vectors.}	Another way of guaranteeing the no-cancellation property is by considering only non-negative gain vectors. The development above again shows that we get logarithmic regret in this case. We remark that the assumption of non-negative gains does not preclude $\|s_t\|$ from growing sublinearly, so this is orthogonal to the assumption in the previous theorem. This was stated in the introduction, we restated it here for convenience. 	
	
	\begin{manualtheorem}{\ref{thm:FTLnew}}
		\thmFTLnew
	\end{manualtheorem}

%	\begin{thm} \label{thm:FTLnew}
%		Consider the OLO problem with playing set $K$ and gain set $\cG$. If $K$ is $\lambda$-strongly convex with respect to a norm $\|\cdot\|$ and all vectors $\cG$ are non-negative,\footnote{That is, $\cG \subseteq \R^d_+$. We note that the proof directly generalizes to the case when $\R^d_+$ is replaced by an arbitrary pointed cone.} then FTL has regret at most 
%		%
%		\begin{align*}
%		\frac{C \cdot M}{\lambda} \cdot \log T,
%		\end{align*}
%		where $M := \max_{g \in \cG} \|g\|$ and $C$ only depends on $\|.\|$.
%	\end{thm}
	
	\begin{proof}
	Since the gain vectors $g_t$ are non-negative, we can assume $\|s_t\| > 0$ for all $t$, otherwise we can just ignore the initial time steps with $g_t = 0$. The idea now is to reduce the analysis to the 1-dimensional case in order to capture more easily the property of no cancellations; for that, we will approximate $\|.\|$ over $\R^d_+$ by a linear function. 
	
	Let $e^i \in \R^d$ denote the $i$th canonical vector, and define the vector $u \in \R^d$ with coordinates $u_i := \|e^i\|$. Define then the linear function $f(x) := \ip{u}{x}$. Notice that $\|x\| \le f(x)$ for all non-negative $x$: by triangle inequality $\|x\| \le \sum_i \|x_i e^i\| = \sum_i x_i u_i = f(x)$. In addition, defining $C := \max\{f(x) : x \in \R^d_+,\, \|x\| = 1\}$, we have $f(x) = \|x\|\cdot f\big(\frac{x}{\|x\|}\big) \le C \cdot \|x\|$ for all $x \in \R^d_+$. Thus, we have the two-sided bound
		\begin{align*}
			\forall x \in \R^d_+,~~\|x\| \le f(x) \le C\cdot \|x\|. 
		\end{align*}

	Employing Lemma \ref{lemma:regretLip} with these bounds, and using the linearity of $f$, the regret of FTL over the gain vectors $g_t$'s is at most
		\begin{align}
		\frac{C}{2\lambda} \sum_{t=1}^T \frac{f(g_t)^2}{f(s_t)} = \frac{C}{2\lambda} \sum_{t=1}^T \frac{f(g_t)^2}{f(g_1) + \ldots + f(g_t)}. \label{eq:FTLpos}
		\end{align}
	To upper bound the right-hand side, we employ the following estimate, which is proved in the appendix. 
	
	\begin{lemma} \label{lemma:logEstimate}
		Let $a_1,\ldots,a_T$ be numbers in $[0,A]$, and let $b_t := a_1 + \ldots + a_t$. Then $\sum_{t = 1}^T \frac{a_t^2}{b_t} \le 5 A \log T.$
	\end{lemma}
	
	Because $f(g_t) \ge 0$ and $f(g_t) \le C \|g_t\| \le C M$ (since by assumption $\|g_t\| \le M$), the previous lemma shows that the right-hand side of \eqref{eq:FTLpos} is at most $\frac{5 C^2 M}{2 \lambda} \log T$. By redefining $C$ we obtain the desired regret bound for FTL, thus concluding the proof.
	\end{proof}

%#########################################################
%#########################################################
%#########################################################
%#########################################################

	\section{Making a Convex Body Curved} \label{sec:curving}

	Consider an arbitrary convex body $K \in \Ko^d$. Our goal in this section is to obtain a set $K_t$ that is strongly convex with respect to itself, that approximates $K$ in the sense of $K_t \subseteq K \subsetsim K_t$, and that can be efficiently optimized over, proving Theorem~\ref{thm:curving}.

	\subsection{A First Attempt}
	
	Let $B(r) \subseteq K$ and $B(R) \supseteq K$ be respectively inscribed and circumscribed balls for $K$. Recall that intuitively a set is strong convex if its boundary does not have flat parts.\footnote{See \cite{curvedFTL} for a formal connection between strong convexity of a set and the curvature of its boundary seen as a Riemannian manifold.}
	
	On one hand, $K$ the is perfect approximation to itself but may not be strongly convex at all; on the other, as we just saw $B(r)$ is $\frac{1}{8}$-strongly convex with respect to itself but (typically) gives a poor approximation to $K$. The idea is to tradeoff these extremes by taking a ``convex combination'' between $K$ and the inscribed ball $B(r)$. 
	
	The natural attempt would be to consider the convex combination $K'_t := (1-t) K + t B(r)$ for $t \in [0,1]$. This operation is just placing a copy of the ball $B(tr)$ at each point of $(1-t) K$, which intuitively should give a more strongly convex set as $t$ increases. Unfortunately this is not true: if $K = [-1,1]^d$, the set $K_t$ is not strongly convex at all for any value $t \in [0,1)$, see Figure \ref{fig:curving}.a. This is because this operation softens the \textbf{corners} of $K$ instead of curving its flat faces. 
	
\begin{figure}[htp]
	\centering		
\includegraphics[width=0.8\textwidth]{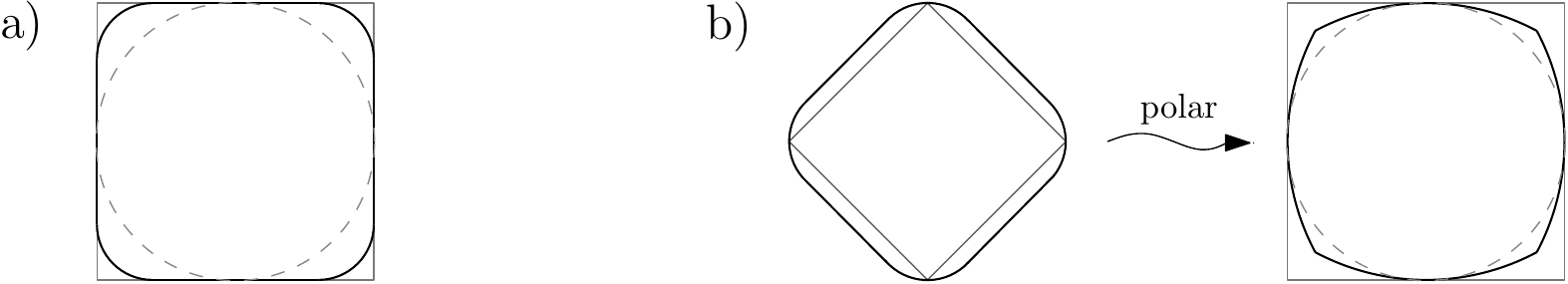}
	\vspace{-5pt}
	\caption{\small a) $K$ being the square, with the black object depicting $(1-t) K + t B$. Notice this object is not strongly convex (or 2-convex). b) The rotated squared is the polar $K^\circ$, and the object in black depicts $(1-t) K^\circ + t B^\circ$. The polar $((1-t) K^\circ + t B^\circ)^\circ$ of this resulting object is shown in black to the right, which is now strongly convex.}
	\label{fig:curving}
\end{figure}
	
	But it is known that polarity maps ``faces'' of the set to ``corners'' of its polar, and vice-versa (Corollary 2.14 of~\citep{ziegler} makes this precise for polytopes). Thus, we should soften the vertices of the \textbf{polar} to obtain the desired effect in the original set. More precisely, we can pass to the polar, take a convex combination with the polar of $B(r)$, and take the polar of the resulting object to get back to the original space:
	\begin{align}
	K''_t := ((1-t) K^\circ + t B(r)^\circ)^{\circ}, \label{eq:polarSum}
	\end{align} for $t \in [0,1]$; see Figure \ref{fig:curving}.b. Indeed, with a careful analysis one can show that $K''_t$ is strongly convex and $K''_t \subseteq K \subsetsim K''_t$ (with the approximation improving as $t \rightarrow 0$). However, we get a greatly simplified analysis by working with a different construction.

%##########################################################
%############################################################

%##########################################################
%############################################################
	
	\subsection{Construction via $L_2$ Addition}	
	
	The idea is to replace the construction given in \eqref{eq:polarSum} by one with a more ``functional'' flavor that gives a clean expression for the its gauge function $\|\cdot\|_{K_t}$. Since Lemma \ref{lemma:SCUC} gives the equivalence between 2-convexity of $\|\cdot\|_{K_t}$ and strong convexity of $K_t$, we will be in good shape for controlling the latter. 

	For this, we replace the Minkowski sum in our previous attempt by the so-called \emph{$L_2$ addition}~\citep{firey,lutwak2}. Given two convex bodies $A,B \in \Ko^d$, their $L_2$ addition $A \oplus B$ is the convex body whose support function satisfies
	\begin{align*}
		\sigma_{A \oplus B}(\cdot)^2 = \sigma_A(\cdot)^2 + \sigma_B(\cdot)^2. 
	\end{align*}
	We then define our desired approximation of $K$ as the set 
	\begin{align*}
		K_t := \left(\big(\sqrt{1-t^2}\,K^\circ\big) \oplus \big(t B(r)^\circ\big)\right)^\circ .
	\end{align*}
	
	To have a more transparent version of this definition, by involution of polarity (Lemma \ref{lemma:polar2} Item 1), the polar of $K_t$ satisfies $K_t^\circ = (\sqrt{1-t^2}\, K^\circ) \oplus (t B(r)^\circ)$ and hence 	
	\begin{align*}
	\sigma_{K_t^\circ}(\cdot)^2 = \sigma_{\sqrt{1-t^2} \,K^\circ}(\cdot)^2 + \sigma_{t B(r)^\circ}(\cdot)^2 &= \left(\sqrt{1-t^2} \cdot \sigma_{K^\circ}(\cdot)\right)^2 + \left(t\cdot \sigma_{B(r)^\circ}(\cdot)\right)^2 \\
	&= (1-t^2)\,\sigma_{K^\circ}(\cdot)^2 + t^2\, \sigma_{B(r)^\circ}(\cdot)^2, 
	\end{align*}
	where in the second equation we used Lemma \ref{lemma:gaugeSupp} Item 5. Moreover, using that the support function of the polar is the gauge of the ``primal'' (Lemma \ref{lemma:polar2} Item 3), we see that $K_t$ is the convex body satisfying
	\begin{align}
	\|\cdot\|_{K_t}^2 = (1-t^2)\,\|\cdot\|^2_K + t^2\,\|\cdot\|^2_{B(r)}. \label{eq:betterDef}
	\end{align}
	
	With this functional perspective we are in good shape for analyzing the properties of $K_t$ and proving Theorem \ref{thm:curving}.
	
%##########################################################
%##########################################################

	\subsection{Proof of Theorem \ref{thm:curving}}
	
	\paragraph{Approximation.} We first argue that $K_t$ is still sandwiched $B(r) \subseteq K_t \subseteq K$. Since $K$ contains the ball $B(r)$, Lemma~\ref{lemma:gaugeSupp} Item~4 gives that $\|\cdot\|_{B(r)} \ge \|\cdot\|_K$. So using \eqref{eq:betterDef} we see that $\|\cdot\|_K \le   \|\cdot\|_{K_t} \le \|\cdot\|_{B(r)}$. The same lemma then implies that $B(r) \subseteq K_t \subseteq K$. 
	
	To see that $K \subseteq \sqrt{1 + ((\frac{R}{r})^2 - 1) t^2}\cdot K_t$, notice that the inclusion $B(r) = \frac{r}{R} B(R) \supseteq \frac{r}{R} K$, together with Lemma \ref{lemma:gaugeSupp} Items 4 and 5, implies that $\|\cdot\|_{B(r)} \le \frac{R}{r} \|\cdot\|_K$, and hence \eqref{eq:betterDef} gives $$\|\cdot\|_{K_t} \le \|\cdot\|_K \cdot \sqrt{1 + \bigg(\bigg(\frac{R}{r}\bigg)^2 - 1\bigg)t^2};$$ the same lemma then give the desired containment. This proves the ``approximation'' part of Theorem \ref{thm:curving}.

	\paragraph{Curvature.} Given the equivalence of strong convexity and 2-convexity of Lemma \ref{lemma:SCUC}, it suffices to show that $\|\cdot\|_{K_t}$ is 2-convex with modulus $\frac{t^2}{8}$. So consider $x,y$ with $\|x\|_{K_t}, \|y\|_{K_t} \le 1$; we want to show that 
	\begin{align}
	\left\|\frac{x+y}{2}\right\|_{K_t} \le 1 - \frac{t^2}{8} \|x-y\|_{K_t}^2. \label{eq:want2conv}
	\end{align}
	First, observe that the function $\|\cdot\|_{K_t}^2$ is convex:
	this follows because it is the composition of the convex function $\|\cdot\|_{K_t}$ (use Lemma \ref{lemma:gaugeSupp} Items 2 and 3 to observe this convexity) and the increasing convex function $z \mapsto z^2$ (see for example Section 3.2.4 of~\citep{Boyd}). Using again the fact $\|\cdot\|_{B(r)} = \frac{1}{r} \|\cdot\|_2$ (Lemma \ref{lemma:gaugeSupp} Item 5), we have 
	\begin{align*}
		\left\|\frac{x+y}{2}\right\|_{K_t}^2 & \stackrel{\textrm{\scriptsize \eqref{eq:betterDef}}}{=} (1-t^2) \left\|\frac{x+y}{2}\right\|_{K}^2 + \frac{t^2}{r^2} \left\|\frac{x+y}{2}\right\|_2^2\\
		& \stackrel{\tiny\textrm{\tiny conv.}}{\le} (1-t^2) \left(\frac{\|x\|_K^2}{2} + \frac{\|y\|_K^2}{2} \right) + \frac{t^2}{r^2}  \left\|\frac{x+y}{2}\right\|_2^2\\
		& \stackrel{\textrm{\tiny parallel.}}{=} (1-t^2) \left(\frac{\|x\|_K^2}{2} + \frac{\|y\|_K^2}{2} \right) + \frac{t^2}{r^2} \left(\frac{\|x\|^2_2}{2} + \frac{\|y\|^2_2}{2} - \frac{\|x - y\|^2_2}{4}\right)\\
		& = \frac{\|x\|_{K_t}^2}{2} + \frac{\|y\|_{K_t}^2}{2} - \frac{t^2}{4r^2} \frac{\|x - y\|^2_2}{4} \\
		& \le 1 -  \frac{t^2}{4r^2} \|x - y\|^2_2 ~=~ 1 - \frac{t^2}{4} \|x - y\|^2_{B(r)} ~\le~ 1 - \frac{t^2}{4} \|x - y\|^2_{K_t},
	\end{align*}
	where in the first inequality we used convexity of $\|\cdot\|_{K_t}^2$, the next equation uses the parallelogram identity, the second inequality uses the assumption $\|x\|_{K_t}, \|y\|_{K_t} \le 1$, and the last inequality uses $\|\cdot\|_{K_t} \le \|\cdot\|_{B(r)}$, proved in the ``approximation'' part. Finally, since $\sqrt{1 - \alpha} \le 1 - \frac{\alpha}{2}$ for all $\alpha \le 1$, taking square roots on the last displayed inequality proves \eqref{eq:want2conv}.
	
%##########################################################	
	
	\paragraph{Efficiency.} It is not immediately clear that we can optimize a linear function over $K_t$ given access to an optimization (or membership) oracle for $K$. First, let us recall the standard definition of weak optimization~\citep{GLS}.
	 
	\begin{definition}[Weak optimization problem] \label{def:weakOpt}
		Given $c \in \R^d$, a convex set $A \subseteq \R^n$, and a precision parameter $\delta > 0$, either:
		\vspace{-4pt}
		\begin{enumerate}
			\itemsep0em
			\item Output that $A - B(\delta)$ is empty
			\item Return a point $\bar{x} \in A + B(\delta)$ such that $$\ip{c}{\bar{x}} \ge \max_{x \in A - B(\delta)} \ip{c}{x} - \delta.$$
		\end{enumerate} 
	\end{definition}
	
	We also recall the following result on the equivalence of weak optimization of a body and its polar (for example, chain together Theorem 4.4.7, Theorem 4.2.2, Lemma 4.4.2, and Corollary 4.2.7 of \cite{GLS}). 
	
	\begin{thm} \label{thm:equivOpt}
		Let $A$ be a convex body satisfying $B(r) \subseteq A \subseteq B(R)$. Then, there is an algorithm that, given access to weak optimization oracles over $A$, solves the weak optimization problem over $A^\circ$ in time polynomial in $r,R,$ and $\delta$.	
	\end{thm}

%	First, recall that, via the Ellipsoid method, optimizing a linear function over a convex set $A$ is polynomial-time equivalent to separating inequalities for it (i.e., given $\bar{y}$, finding a hyperplane separating $\bar{y}$ from $A$  if $\bar{y} \notin K_t$)~\citep{GLS}.\footnote{In order to simplify the exposition we ignore that fact that optimization/separation can only be performed within an additive error. But it is easy to verify that assuming an $\delta$-error separation oracle for $K$, the reduction in the sequel allows one to optimize a linear function over $K_t$ within error $O(\delta)$ in time $poly(d,1/\delta,R/r)$.} Moreover, since by involution of polarity we have $A = (A^{\circ})^\circ = \{y : \ip{y}{x} \le 1, ~\forall x \in A^\circ\}$, to solve the separation problem for $A$ it suffices to find $x \in A^\circ$ such that $\ip{\bar{y}}{x} > 1$, so it suffices to maximize a linear function $x \mapsto \ip{\bar{y}}{x}$ over the polar $A^\circ$.
	
%	$K_t$ is polynomial-time equivalent to separating inequalities for it (i.e., given $\bar{y}$, finding a hyperplane separating $\bar{y}$ from $K_t$ if $\bar{y} \notin K_t$)~\citep{GLS}. Moreover, since by involution of polarity we have $K_t = (K_t^{\circ})^\circ = \{y : \ip{y}{x} \le 1, ~\forall x \in K_t^\circ\}$, to solve this separation problem it suffices to find $x \in K^\circ_t$ such that $\ip{\bar{y}}{x} > 1$. Thus, it suffices to maximize a linear function $\ip{\bar{y}}{x}$ over $x \in K_t^\circ$. 
	
	Given this equivalence and the involution of polarity $K_t = (K_t^\circ)^\circ$, in order to weakly optimize over $K_t$ it suffices to be able to weakly optimize over its polar $K_t^\circ$. To do that, we will need a characterization of the the $L_2$ addition by~\cite{lutwak2}, which when applied to $K_t^\circ$ gives the following (to simplify the notation, let $U := \sqrt{1-t^2}\, K^\circ$ and $V := t B(r)^\circ$): 
	\begin{align*}
		K_t^\circ = \left\{(1-\alpha)^{1/2} u + \alpha^{1/2} v \,:\, u \in U,\ v \in V,\ \alpha \in [0,1] \right\}.
	\end{align*} 
	Thus, given $\bar{y}$, maximizing $\ip{\bar{y}}{x}$ over $x \in K^\circ_t$ is equivalent to the following optimization problem:
	\begin{align*}
		\max~~& (1 - \alpha)^{1/2} \ip{\bar{y}}{u} + \alpha^{1/2} \ip{\bar{y}}{v}\\
		\textrm{s.t.}~~& u \in U,\ v \in V,\, \alpha \in [0,1].
	\end{align*}
	Given the decomposability of this problem, we can do this in polynomial time as follows:
	\begin{enumerate}
		\item First weakly maximize $\ip{\bar{y}}{u}$ over $u \in U$, obtaining an almost optimal solution $u^*$. Again, by Theorem \ref{thm:equivOpt} this is equivalent to weakly optimizing over the polar $U^\circ = \frac{1}{\sqrt{1-t^2}} K$, which (since $t$ is fixed) is equivalent to weakly optimizing over $K$, which we assumed we have an oracle for. 
		
		\item Then maximize $\ip{\bar{y}}{v}$ over $v \in V$, obtaining the optimal solution $v^*$. Notice that $V = t B(r)^\circ = t B(\frac{1}{r})$ (Lemma \ref{lemma:polar2} Item 4), so it is just the Euclidean ball of radius $\frac{t}{r}$. Thus, we explicitly have the maximizer $v^* = \frac{t}{r} \frac{\bar{y}}{\|\bar{y}\|}$. 
		
		\item Finally, weakly maximize $\phi(\alpha):= (1-\alpha)^{1/2} \ip{\bar{y}}{u^*} + \alpha^{1/2} \ip{\bar{y}}{v^*}$ over $\alpha \in [0,1]$, obtaining an almost optimal solution $\alpha^*$. We claim that $\phi$ is concave in $[0,1]$. To see this, notice that since $U$ has the origin in its interior, the optimality of $u^*$ gives that $\ip{\bar{y}}{u^*} \ge 0$, and the same is true for $v^*$. Then one can easily check that the second derivative of $\phi$ is negative in $[0,1)$, thus giving its concavity over $[0,1]$ (also notice that $\phi$ is continuous at $1$). Thus, we can weakly optimize $\phi$ in polynomial time (see for example Theorem 4.3.13 of \citep{GLS}).
	\end{enumerate}
	
	Putting all these elements together, we can weakly optimize over $K_t$ in polynomial time using a weak optimization oracle for $K$. With this, we conclude the proof of Theorem \ref{thm:curving}. 

%############################################################
%############################################################
%############################################################

	\subsection{Application: OLO with Hints} \label{app:application}
	
 Dekel et al.	\cite{nika} considered Online Linear Optimization with the addition of \emph{hints} (see \citep{megiddoHints,rakhlinPredictable} for other notions of hints). Here, at time $t$ the algorithm receives a ``hint'' vector $h_t$ in the Euclidean sphere which is guaranteed to satisfy $\ip{h_t}{g_t} \ge \alpha \|g_t\|$, namely it is correlated with the unknown gain vector $g_t$. The hint $h_t$ can be thought of a prediction of $g_t$, which may be available when there is additional structure in the problem, e.g., the gain vectors are not adversarial. Dekel et al. showed that when the playing set $K$ is curved and hints are present, one can obtain improved $O(\log T)$ regret, instead of the standard $O(\sqrt{T})$. 
	 	
	\begin{thm}[\cite{nika}] \label{thm:nika}
		Consider the Online Linear Optimization problem with hints. Suppose $K \subseteq \R^d$ is centrally symmetric and $B(r) \subseteq K \subseteq B(R)$. Suppose further that its gauge $\|\cdot\|_K$ is $(2,D)$-convex. Finally, let $G := \max_{g \in \cG} \|g\|_2$ and suppose the unit-norm hints satisfy $\ip{h_t}{g_t} \ge \alpha \|g_t\|_2$ for all $t$. Then there is an algorithm with regret at most $$\frac{d\,G\,R^2}{\alpha\,D\,r}\, O(\log T).$$ 
	\end{thm}
	
	The curvature of $K$ is crucial for this improved regret even in the presence of hints, since otherwise there is still a lower bound of order $\sqrt{T}$~\citep{nika}.
	
	\medskip 
	We can use our ``curving convex bodies'' Theorem \ref{thm:curving} to extend this result to \textbf{all} playing sets $K$, obtaining additive regret $O(\log T)$ at the expense of an additional multiplicative regret.
	
	\begin{thm} \label{thm:hints}
		Consider the online linear optimization problem with hints on a centrally symmetric body $K \in \R^d$ with $B(r) \subseteq K \subseteq B(R)$. Then for every $\e > 0$ there is an algorithm with total reward at least 
		\begin{align*}
			(1-\e) \OPT - \frac{d G R^4}{\alpha r^3} \cdot O\left(\frac{\log T}{\e}\right),
		\end{align*}
		where $\OPT = \max_{x \in K} \sum_{t = 1}^T \ip{x}{g_t}$ and $G = \max_{g \in \cG} \|g\|_2$.
	\end{thm}
		
	\begin{proof}		
	For the algorithm, simply run the algorithm of Theorem \ref{thm:nika} over the set $K_t$ instead of the original one $K$, setting $t$ to satisfy $t^2 = \frac{2\e}{(\frac{R}{r})^2-1}$. Since $K_t \subseteq K$, this strategy only plays feasible actions. 
	
	For its guarantee, the second item of Theorem~\ref{thm:curving} plus Lemma \ref{lemma:SCUC} gives that the modulus of 2-convexity of $\|\cdot\|_{K_t}$ is at least $\frac{\e r^2}{4 R^2}$. Thus, applying Theorem \ref{thm:nika} to the game with playing set $K_t$, the total gain of the algorithm is at least $$\OPT_{K_t} - \frac{dGR^4}{\alpha r^3} \cdot O\left(\frac{\log T}{\e}\right),$$ where $\OPT_{K_t} := \max_{x \in K_t} \sum_{t =1}^T \ip{x}{g_t}$ is the optimal fixed action in $K_t$. Moreover, using the definition of $\e$ we also have the containment $K \subseteq \sqrt{1 + 2\e}\, K_t$, and using again $\sqrt{1 + x} \le 1+ \frac{x}{2}$, this implies $K \subseteq (1+\e)\, K_t$. Then the optimal solutions of the games on $K$ and $K_t$ satisfy $\OPT_{K_t} \ge \frac{1}{1+\e} \OPT \ge (1-\e) \OPT$, since scaling the optimal solution achieving $\OPT$ by $\frac{1}{1 + \e}$ yields a feasible solution for the game on $K_t$. Employing this on the displayed inequality shows that the algorithm has gains at least $(1-\e) \OPT - \frac{d G R^4}{\alpha r^3} \cdot O(\frac{\log T}{\e})$. This conclude the proof of the theorem. 
	\end{proof}

	Notice that for general playing sets $K$ the $\Omega(\sqrt{T})$ regret lower bound of~\cite{nika} implies that the multiplicative loss is required, and the standard $\Omega(\sqrt{T})$ lower bound for the game without hints  shows that such $O(\log T)$ regret is not possible in the absence of hints, even if multiplicative losses are allowed (e.g., $K = [-1,1]$ and the losses are i.i.d. $g_t \in_{R} \{-1,1\}$; the expected loss of any algorithm is 0, while the best action in hindsight has expected loss $-\Omega(\sqrt{T})$).

%##########################################################
%##########################################################
%##########################################################
%##########################################################

\section*{Acknowledgements}

We thank Jacob Abernethy for discussions on the topics of this paper, and Willie Wong for the example in Section \ref{sec:smoothSmooth} of a norm that is 2-smooth but not twice differentiable.

\bibliographystyle{siam}
\bibliography{../online-lp-short}

%##########################################################
%##########################################################
%##########################################################
%##########################################################

\appendix

\section{Non-midpoint Strong Convexity} \label{app:nonMidpoint}

The following definition of curvature was used in~\citep{garberHazan}. 

	\begin{definition}[Non-midpoint Strongly Convex Sets] \label{def:nmpSC}
		Consider a convex body $C$ with the origin in its interior. 
		The convex body $K$ is \textbf{\emph{$\lambda$-non-midpoint strongly convex}} with respect to $C$ if for every $x,y \in K$ and every $z = \mu x + (1-\mu)y$ with $\mu \in [0,1]$ we have the containment $$z + 4\lambda \, \mu \, (1-\mu) \, \|x-y\|_C^2 \cdot C\subseteq K.$$
	\end{definition}
	
	It is clear every non-midpoint strongly convex set is strongly convex. The next lemma shows the other direction.
	
	\begin{lemma}
		$\lambda$-Strong convexity implies $\frac{\lambda}{2}$-non-midpoint strong convexity. 
	\end{lemma}
	
	\begin{proof}
		Consider a $\lambda$-strongly convex set $K$ with respect to $C$. Consider any pair of points $x,y \in K$ and $z = \mu x + (1-\mu) y$ with $\mu \in [0,1]$. Let $\e := \|x - y\|_C$. By symmetry, assume without loss of generality that $\mu \le \frac{1}{2}$. Let $m = \frac{x + y}{2}$ be the midpoint of $x$ and $y$.	By assumption, $m + \lambda \e^2 C \subseteq K$. 
		
		We claim that the set $z + 2 \lambda \mu \e^2 C$ is contained in the convex hull of $m + \lambda \e^2 C$ and $y$; convexity of $K$ implies that $K$ also contains this set, which would conclude the proof. To prove the claim, note we can write $z = y + 2\mu (m - y)$, which equals $(1 - 2 \mu) y + 2 \mu m$. The convex combination between $y$ and $m + \lambda \e^2 C$ with coefficient $(1- 2\mu)$ (recall that by assumption $2\mu \in [0,1]$) is precisely $(1-2\mu) y + 2 \mu (m + \lambda \e^2 C) = z + 2 \mu \lambda \e^2 C$, hence the right-hand side belongs to the convex hull of $y$ and $m + \lambda \e^2 C$, as desired. This concludes the proof.   
	\end{proof}

%##########################################################
%##########################################################
%##########################################################
%##########################################################	
	
	\section{Proof of Lemma \ref{lemma:logEstimate}}
	
		Let $t_0$ be the first index where $a_{t_0} > \frac{A}{T}$ (let $t_0 = T$ if no such exists). Since $\frac{a_t}{b_t} \le 1$ and $\frac{a_t^2}{b_t} \le A$, we have 
		\begin{align}
			\sum_{t \le t_0} \frac{a_t^2}{b_t} \le A + \sum_{t < t_0} \frac{a_t^2}{b_t} \le A + \frac{A}{T} \sum_{t < t_0} \frac{a_t}{b_t} \le 2A.  \label{eq:trick1}
		\end{align}
		To upper bound the contribution of the indices $t > t_0$, notice by concavity of $\log$ that $x \le \frac{\log (1 + x)}{\log 2} \le 1.5 \log(1+x)$ for $x \in [0,1]$. Applying this inequality for $t > t_0 \ge 1$ we get 
		\begin{align*}
			\frac{a_t}{b_t} \le 1.5 \log\left(1 + \frac{a_t}{b_t}\right) \le  1.5 \log\left(1 + \frac{a_t}{b_{t-1}}\right)   = 1.5 \log \left(\frac{b_t}{b_{t-1}}\right). 
		\end{align*}
		Using the upper bound $A$ and adding over $t > t_0$, the $\log$'s telescope and 
		\begin{align}
			\sum_{t = t_0 + 1}^T \frac{a_t^2}{b_t} \le A \sum_{t = t_0 + 1}^T \frac{a_t}{b_t} \le 1.5 A \log \left(\frac{b_T}{b_{t_0}}\right) \le 3 A \log T, \label{eq:trick2}
		\end{align}
		where the last inequality uses $b_T \le A \cdot T$ and $b_{t_0} \ge a_{t_0} > \frac{A}{T}$. Adding the bounds \eqref{eq:trick1} and \eqref{eq:trick2} concludes the proof.

\end{document}